\documentclass[a4paper,onecolumn,11pt,unpublished]{quantumarticle}
\pdfoutput=1
\usepackage[utf8]{inputenc}
\usepackage[english]{babel}
\usepackage[T1]{fontenc}

\usepackage{tikz}
\usepackage{lipsum}

\usepackage{amsmath,amssymb,amsthm,mathtools}
\usepackage{braket}
\usepackage{bm}
\usepackage{microtype}
\usepackage{enumitem}
\usepackage{thm-restate}
\usepackage[numbers,sort&compress]{natbib}
\usepackage{graphicx}

\newtheorem{theorem}{Theorem}[section]

\newtheorem{lemma}[theorem]{Lemma}

\newtheorem{definition}[theorem]{Definition}

\newcommand{\C}{\mathbb{C}}
\newcommand{\E}{\mathbf{E}}
\newcommand{\Prb}{\mathbf{Pr}}
\newcommand{\cH}{\mathcal{H}}

\newcommand{\cS}{\mathcal{S}}
\newcommand{\cX}{\mathcal{X}}
\newcommand{\cY}{\mathcal{Y}}
\newcommand{\cC}{\mathcal{C}}
\newcommand{\cD}{\mathcal{D}}
\newcommand{\cB}{\mathcal{B}}
\newcommand{\cM}{\mathcal{M}}
\newcommand{\cP}{\mathcal{P}}
\newcommand{\Tr}{\operatorname{Tr}}
\newcommand{\rank}{\operatorname{rank}}

\newcommand{\err}{\operatorname{err}}
\newcommand{\opt}{\operatorname{opt}}
\newcommand{\vc}{\operatorname{VCdim}}
\newcommand{\Dtr}{D_{\mathrm{tr}}}
\newcommand{\id}{\operatorname{id}}

\newcommand{\Gr}{\operatorname{Gr}}
\newcommand{\ketbra}[2]{\ket{#1}\!\bra{#2}}
\newcommand{\proj}[1]{\ketbra{#1}{#1}}
\newcommand{\HC}{\mathsf{HC}}

\newcommand{\norm}[1]{\left\lVert#1\right\rVert}

\usepackage[colorlinks=true,linkcolor=blue,citecolor=blue,urlcolor=blue]{hyperref}

\begin{document}

\title{Advantage of Sample Complexity in Quantum PAC Learning Requires Inverse Access to State-Preparation Unitaries}

\author{Natsuto Isogai}
\email{natsutoisogai@g.ecc.u-tokyo.ac.jp}
\affiliation{Department of Physics, Graduate School of Science, The University of Tokyo, Tokyo, Japan}
\orcid{0009-0002-5496-5871}
\author{Satoshi Yoshida}
\email{satoshiyoshida.phys@gmail.com}
\orcid{0000-0002-0521-5209}
\affiliation{Department of Physics, Graduate School of Science, The University of Tokyo, Tokyo, Japan}
\author{Mio Murao}
\email{murao@phys.s.u-tokyo.ac.jp}
\affiliation{Department of Physics, Graduate School of Science, The University of Tokyo, Tokyo, Japan}
\orcid{0000-0001-7861-1774}

\maketitle

\begin{abstract}
Whether quantum computation can reduce the amount of data required for conventional machine learning is a central problem in quantum machine learning. In probably approximately correct (PAC) learning, a learner receives training examples sampled from an unknown probability distribution. Quantum PAC learning extends this framework by allowing access to quantum examples, which are quantum states that coherently encode this distribution, or to unitaries that prepare quantum examples, with or without inverse access. When only copies of quantum examples are available, the optimal worst-case sample complexity asymptotically matches that of classical learning. In contrast, access to both a preparation unitary and its inverse can improve the accuracy dependence of query complexity. However, it has remained open whether access to the preparation unitary alone suffices for this improvement.

In this work, we resolve this question by determining the optimal query complexity with forward-only access, taking the worst case over compatible preparation unitaries and their dimensions. In both realizable and agnostic PAC learning, this complexity matches, up to constant factors, the optimal sample complexity achievable with classical examples or copies of quantum examples. To prove these results, we construct, for any algorithm with forward-only access, a simulator that approximates the algorithm's Haar-averaged output state using as many copies of the prepared state as the algorithm makes queries. These results show that inverse access is essential for the known improvement in accuracy dependence in this worst-case setting. Beyond PAC learning, our reduction provides a general tool for transferring lower bounds for state copies to query lower bounds for forward-only access to the corresponding preparation unitaries.
\end{abstract}

\setcounter{tocdepth}{2}
\tableofcontents

\section{Introduction}\label{sec:introduction}
Machine learning is used in a wide range of technologies that support our everyday lives.
Quantum machine learning (QML) studies whether quantum computation can reduce the resources required for learning~\cite{wittek2014,schuld2021machine}.
Identifying such advantages requires comparing the resources classical and quantum methods need to accomplish the same task under matched conditions.
In particular, the outcome of this comparison can depend on the data and operations available to the learner.

A standard framework for such comparisons is \textit{probably approximately correct (PAC) learning}~\cite{valiant1984theory,kearns1990computational,kearns1994introduction}.
PAC learning asks whether a learner can construct a prediction rule from data drawn from an unknown distribution and achieve a prescribed accuracy with high probability.
The explicit conditions for successful learning in PAC learning enable comparisons between models that allow different operations.
Although reducing computation time through quantum computation is an important research direction~\cite{servedio2004,gyurik2023establishinglearningseparationsclassical,gyurik2024exponentialseparationsclassicalquantum,Yamasaki2026,Molteni2026}, determining the amount of data required for learning is also a central problem in learning theory.
Sample complexity quantifies how much data is needed to learn to a desired accuracy and confidence, and measures the statistical difficulty of a learning task.
Understanding the required amount of data is important both for determining how much data to provide and for designing and evaluating computationally efficient learning algorithms.

To investigate this question, researchers have studied models in which the distribution of classical learning data is encoded in quantum data given by a quantum state that is supplied to the learner~\cite{10.1145/225298.225312,servedio2004,Atici2005,Atici2007,arunachalam2018optimal,arunachalam2017survey}.
Measuring the quantum states yields corresponding classical data, while a quantum learner can also coherently process the states jointly before measurement.
This raised the possibility that quantum data could enable learning from fewer samples than classical data.
However, in a model where the learner only receives identical copies of a quantum state, there is no quantum advantage in the worst-case sample complexity of general PAC learning tasks in the asymptotic regime~\cite{arunachalam2018optimal}.

Receiving quantum data differs from having multiple accesses to a state-preparation unitary.
Salmon et al.\ showed that, when both this unitary and its inverse are available, amplitude amplification enables PAC learning with a query complexity lower than the sample complexity required~\cite{salmon2024provable}.
In particular, the dependence of the required number of queries on the desired learning accuracy admits a quadratic improvement, up to logarithmic factors.

Access to a state-preparation unitary, however, does not necessarily provide access to its inverse.
Methods for implementing the inverse of an unknown unitary using forward access are known, but they require additional $\Theta(N^2)$ overheads for $N$-dimensional unitary queries~\cite{PhysRevLett.131.120602,chen2025quantumalgorithmreversingunknown,drp2-rzzw,10.1145/3798129.3800847}.
Therefore, this does not imply that forward-only access offers the same learning advantages as access that includes the inverse.
Tang and Wright established separations between forward-only access and access that also permits inverse queries for tasks involving amplitude amplification and amplitude estimation~\cite{tang2026inverse}.
However, PAC learners need not use amplitude amplification, so these results do not directly yield learning lower bounds.
On the other hand, Chen showed that forward-only access to a state-preparation unitary achieves a query complexity with a better dependence on the target accuracy than the sample complexity achievable using only copies of the state for quantum state estimation~\cite{chen2025inversefreequantumstateestimation}.
This demonstrates that forward access has power beyond state copies, but the required number of queries also depends strongly on the dimension of the state space.
These results alone therefore do not determine the power of forward access in quantum PAC learning.

In this work, we require the learner to succeed for every state-preparation unitary that generates the same quantum data.
Under this requirement, we study whether quantum PAC learning with forward-only access offers an advantage over classical data, and clarify its relationship to learning with inverse access.

\subsection{Main Results}
We establish optimal query lower bounds for both realizable and agnostic quantum PAC learning with forward-only access to state-preparation unitaries.
To state our results, we first describe the learning procedure and the access models being compared (see Section~\ref{sec:pac-learning} for formal definitions).

\paragraph{Learning tasks.}
In PAC learning, a learner repeatedly receives data consisting of an input $x\in\cX$ and a label $y\in\{0,1\}$, and outputs a hypothesis $h:\cX\to\{0,1\}$ for predicting the labels of future inputs.
Such an input--label pair $(x,y)$ is called an \textit{example} in PAC learning.
The PAC learning model has two main settings: \textit{realizable learning} and \textit{agnostic learning}.
In realizable learning, the labels are assumed to follow an unknown concept $c$ from a known concept class $\cC$, so that $y=c(x)$.
On each query, the classical example oracle $\mathrm{PEX}(c,\mathcal{D})$ independently draws $x$ from the unknown input distribution $D$ and returns the pair $(x,c(x))$.
The learner aims to output a hypothesis whose error probability $\Prb_{x\sim D}[h(x)\neq c(x)]$ on a fresh input from the same distribution is at most $\varepsilon$.
We require a hypothesis satisfying this condition to be obtained with probability at least $1-\delta$, over the received data and the randomness of the learner.
Thus, $\varepsilon$ is the allowed prediction error, and $\delta$ bounds the probability that learning fails to meet this requirement.

In agnostic learning, the labels need not follow a concept.
Instead, the input--label pairs follow an arbitrary unknown joint distribution $\cD$.
The classical example oracle $\mathrm{AEX}(\cD)$ returns independent pairs $(x,y)$ drawn from this distribution.
The goal is to output, with probability at least $1-\delta$, a hypothesis whose error exceeds that of the best concept in $\cC$ by at most $\varepsilon$.
In both learning settings, we allow \textit{improper learners}, whose output hypotheses need not belong to $\cC$.

\paragraph{Quantum PAC learning.}
The corresponding quantum example oracles $\mathrm{QPEX}(c,\mathcal{D})$ and $\mathrm{QAEX}(\cD)$ return fresh copies of the respective states
\begin{equation}
\begin{aligned}
    \ket{\psi_{\mathcal{D},c}}
    =\sum_{x\in\cX}\sqrt{\mathcal{D}(x)}\ket{x,c(x)},
    \quad
    \ket{\psi_{\cD}}
    =\sum_{x\in\cX}\sum_{y\in\{0,1\}}
      \sqrt{\cD(x,y)}\ket{x,y}
\end{aligned}
\end{equation}
on each query.
In these states, the probability of each pair is encoded as the squared magnitude of its amplitude.
Measurement in the computational basis produces classical data with the same distribution as $\mathrm{PEX}$ or $\mathrm{AEX}$, while a quantum learner can also apply joint quantum operations to multiple copies before measurement.

Writing these quantum examples as $\ket{\psi}$, we also consider a model that allows queries to a state-preparation unitary $U$ satisfying $U\ket{0}=\ket{\psi}$ for a known initial state.
In this model, the learner can apply the same unknown $U$ to arbitrary query states during learning.
We call access to $U$ alone forward access, and access to both $U$ and its inverse $U^\dagger$ forward-and-inverse access.
We count queries to classical example oracles, copies of quantum examples, or queries to state-preparation unitaries and their inverses as resources, and do not include computation time.

\paragraph{Main theorems.}
Our lower bounds for forward access require the learner to succeed for every state-preparation unitary satisfying $U\ket{0}=\ket{\psi}$.
Since the unitary preparing a given quantum example is not unique, this requirement imposes no particular structure on the action of $U$ on inputs orthogonal to $\ket{0}$.
We call the dimension $N$ of the Hilbert space on which $U$ acts the \textit{ambient dimension}.
The ambient dimension is not fixed, and we take the worst-case query complexity over permitted ambient dimensions and compatible state-preparation unitaries.

Let $d=\vc(\cC)$ be the Vapnik-Chervonenkis (VC) dimension of the concept class.
This is the largest size of a set of points on which $\cC$ realizes every binary labeling (see Definition~\ref{def:vc-dimension} for details).
In the setting above, we show the following two theorems.
\begin{theorem}[Informal version of Theorem~\ref{thm:realizable-main}]
    \label{thm:realizable-intro}
    Let $\cC$ be a Boolean concept class with finite VC dimension $d=\vc(\cC)\geq2$, and let $0<\varepsilon\leq1/32$ and $0<\delta\leq1/4$.
    For realizable PAC learning with forward-only access to state-preparation unitaries, the worst-case query complexity satisfies
    \begin{equation}
        q=\Omega\!\left(
            \frac{d+\log(1/\delta)}{\varepsilon}
        \right).
    \end{equation}
\end{theorem}

\begin{theorem}[Informal version of Theorem~\ref{thm:agnostic-main}]
    \label{thm:agnostic-intro}
    Let $\cC$ be a Boolean concept class with finite VC dimension $d=\vc(\cC)\geq1$, and let $0<\varepsilon\leq1/16$ and $0<\delta\leq1/4$.
    For agnostic PAC learning with forward-only access to state-preparation unitaries, the worst-case query complexity satisfies
    \begin{equation}
        q=\Omega\!\left(
            \frac{d+\log(1/\delta)}{\varepsilon^2}
        \right).
    \end{equation}
\end{theorem}
Matching upper bounds can be achieved by preparing and measuring $U\ket{0}$ on each query and applying a classical learning algorithm with optimal sample complexity to the resulting classical examples~\cite{hanneke2016optimal,a0b266f6-398c-301c-9bea-dc3d41c6ccf7}.
These lower bounds are therefore optimal up to constant factors.

\begin{table*}[t]
\centering
\caption{PAC learning complexity in each access model.
Let $d=\vc(\cC)$.
For the sample-access and forward-only bounds, the realizable setting assumes $d \geq 2$, $0<\varepsilon\leq1/32$, and $0<\delta\leq1/4$, and the agnostic setting assumes $d\geq1$, $0<\varepsilon\leq1/16$, and $0<\delta\leq1/4$.
The forward-only lower bounds take the worst case over permitted ambient dimensions and compatible unitary completions.
Bold entries indicate the new lower bounds established in this paper.
}
\label{tab:pac-oracle-bounds}
\normalsize
\setlength{\tabcolsep}{4pt}
\renewcommand{\arraystretch}{1.5}
\begin{tabular*}{\textwidth}{@{\extracolsep{\fill}}llccc@{}}
\hline
\textbf{Setting} & \textbf{Oracle} & \textbf{Upper bound} & \textbf{Lower bound} & \textbf{Reference} \\[0.4ex]
\hline
\noalign{\vskip 0.8ex}
Realizable
  & $\mathrm{PEX}$ & $O\!\left(\dfrac{d+\log(1/\delta)}{\varepsilon}\right)$ & $\Omega\!\left(\dfrac{d+\log(1/\delta)}{\varepsilon}\right)$
  & \cite{10.1145/76359.76371,EHRENFEUCHT1989247,hanneke2016optimal} \\[1.4ex]
  & $\mathrm{QPEX}$ & $O\!\left(\dfrac{d+\log(1/\delta)}{\varepsilon}\right)$ & $\Omega\!\left(\dfrac{d+\log(1/\delta)}{\varepsilon}\right)$
  & \cite{arunachalam2018optimal} \\[1.4ex]
  & $U$ & $O\!\left(\dfrac{d+\log(1/\delta)}{\varepsilon}\right)$ & {\boldmath$\Omega\!\left(\dfrac{d+\log(1/\delta)}{\varepsilon}\right)$}
  & Theorem~\ref{thm:realizable-intro} \\[1.4ex]
  & $U,\ U^\dagger$
  & $O\!\left(\dfrac{(d+\log(1/\delta))\log^9(1/\varepsilon)}{\sqrt{\varepsilon}}\right)$
  & $\Omega(d/\sqrt{\varepsilon})$
  & \cite{salmon2024provable} \\[1.4ex]
\noalign{\vskip 0.4ex}
\hline
\noalign{\vskip 0.8ex}
Agnostic
  & $\mathrm{AEX}$ & $O\!\left(\dfrac{d+\log(1/\delta)}{\varepsilon^2}\right)$ & $\Omega\!\left(\dfrac{d+\log(1/\delta)}{\varepsilon^2}\right)$
  & \cite{a0b266f6-398c-301c-9bea-dc3d41c6ccf7,vapnik74theory} \\[1.4ex]
  & $\mathrm{QAEX}$ & $O\!\left(\dfrac{d+\log(1/\delta)}{\varepsilon^2}\right)$ & $\Omega\!\left(\dfrac{d+\log(1/\delta)}{\varepsilon^2}\right)$
  & \cite{arunachalam2018optimal} \\[1.4ex]
  & $U$ & $O\!\left(\dfrac{d+\log(1/\delta)}{\varepsilon^2}\right)$ & {\boldmath$\Omega\!\left(\dfrac{d+\log(1/\delta)}{\varepsilon^2}\right)$}
  & Theorem~\ref{thm:agnostic-intro} \\[1.4ex]
\noalign{\vskip 0.4ex}
\hline
\end{tabular*}
\end{table*}

Table~\ref{tab:pac-oracle-bounds} summarizes the known results and our results, where the rows $U$ and $U,U^\dagger$ correspond to forward-only and forward-and-inverse access, respectively.
The two lower bounds in bold are the new results of this paper.
It is known that the optimal sample complexities with classical examples and with copies of quantum examples asymptotically coincide~\cite{arunachalam2018optimal}.
We show that, for both realizable and agnostic learning, forward access to a unitary does not improve this dependence in the worst-case setting.
In realizable learning, however, forward-and-inverse access is known to improve the accuracy dependence from $1/\varepsilon$ to $1/\sqrt{\varepsilon}$, up to logarithmic factors~\cite{salmon2024provable}.
Our results show that forward-only access to a state-preparation unitary is insufficient for this improvement, and that inverse access plays an important role in this case.

\subsection{Method}

We show Theorems~\ref{thm:realizable-intro} and~\ref{thm:agnostic-intro} by transferring known lower bounds for copies of quantum examples to lower bounds for forward access to state-preparation unitaries.
Namely, our approach is to transform a learner using $q$ forward queries into a learner receiving $q$ copies of the same quantum example, while guaranteeing nearly the same success probability.
Such a transformation would make lower bounds on the number of copies required by the learner receiving quantum examples applicable to the number of queries used by the original learner.

The obstacle is that unitary access permits more operations than simply receiving quantum examples.
The relation $U\ket{0}=\ket{\psi}$ allows one copy to be obtained with a single query, but the learner can also apply $U$ to arbitrary inputs orthogonal to $\ket{0}$.
Using information about $\ket{\psi}$ contained in this action might enable learning that cannot be achieved with copies alone.
Therefore, replacing each query with one copy is not sufficient to apply the known copy lower bounds.

In related work, Tang, Wright, and Zhandry showed how to simulate state-preparation access using copies of the prepared state~\cite{tang2026conjugatequerieshelp}.
With sufficient auxiliary space, the output of an algorithm making a total of $q$ queries to a state-preparation unitary and its inverse, averaged over a particular distribution of unitaries, can be reproduced using $O(q^2/\eta)$ copies to within trace distance $\eta$.
However, even at constant approximation error, the copy cost of this reduction grows quadratically with the query count, so its guarantee does not directly transfer known copy lower bounds to query lower bounds of the same order.
We therefore restrict attention to forward-only access and construct a reduction that approximately reproduces $q$ queries using the same number of copies, namely $q$.
Specifically, we choose the unitary Haar-randomly subject to $U\ket{0}=\ket{\psi}$ and consider the output averaged over this distribution.

\paragraph{VC dimension and ambient dimension.}
The key idea is to distinguish the VC dimension from the dimension of the Hilbert space on which the unitary acts.
The optimal sample complexity of PAC learning with classical examples or copies of quantum examples is determined by the VC dimension $d$, the accuracy $\varepsilon$, and the failure probability $\delta$.
In contrast, the resources required for quantum state estimation depend on the dimension of the state space~\cite{BRU1999249,vanapeldoorn2022quantumtomographyusingstatepreparation,chen2025inversefreequantumstateestimation}.
Motivated by this distinction, we increase the ambient dimension $N$ without changing the learning problem, making the remaining action of the unitary harder to exploit.

Specifically, the quantum examples used in the known copy lower bounds lie in a common known subspace $\cS$ spanned by at most $2d$ basis states corresponding to $d$ inputs and binary labels.
Writing its dimension as $r=\dim\cS$, we can append basis states with zero amplitude in the quantum examples and make the ambient dimension $N$ on which the unitary acts arbitrarily large while keeping $r\leq2d$ fixed.
This enlargement changes neither the example distribution nor the concept class, and preserves the known copy lower bounds.
We construct a hard setting for forward access by maintaining $U\ket{0}=\ket{\psi}$ and choosing the remaining action Haar-randomly on the enlarged space.
To show that this choice suppresses any additional advantage, we must account for arbitrary learning strategies.
To this end, we prove the following reduction theorem, which approximates the output state of the learning algorithm itself using copies.

\paragraph{Reduction theorem and its application to PAC lower bounds.}
\begin{theorem}[Informal version of Theorem~\ref{thm:copy-reduction}]\label{thm:main}
Let the unknown unit vector $\ket{\psi}$ belong to a known $r$-dimensional subspace $\cS$ orthogonal to the initial state $\ket{0}$.
On an $N$-dimensional space, choose the remaining action of a unitary $U$ satisfying $U\ket{0}=\ket{\psi}$ Haar-randomly.
For any algorithm making $q$ forward queries to this same $U$, there exists a quantum channel that approximates its Haar-averaged output using $q$ copies of $\ket{\psi}$.
When $q(r-1)<N-1$, the approximation error in trace distance is at most
\begin{equation}\label{eq: closeness-reduction}
    \frac{q(r-1)}{N-1},
\end{equation}
where this channel is independent of the unknown $\ket{\psi}$, and the same channel works for every $\ket{\psi}$ in $\cS$.
\end{theorem}
The theorem applies to arbitrary forward-query algorithms without assuming any particular task or success conditions.
The trace-distance guarantee also bounds the difference in the probability of any success event between the original average output and the simulator output by the same error.
Thus, by fixing $q$ and $r$ and increasing the ambient dimension $N$, we can reproduce the Haar-averaged output obtained with forward access using the same number of copies to arbitrarily small error.
We illustrate Theorem~\ref{thm:main} in Fig.~\ref{fig:reductionmain}.

\begin{figure}
    \centering
    \includegraphics[width=\linewidth]{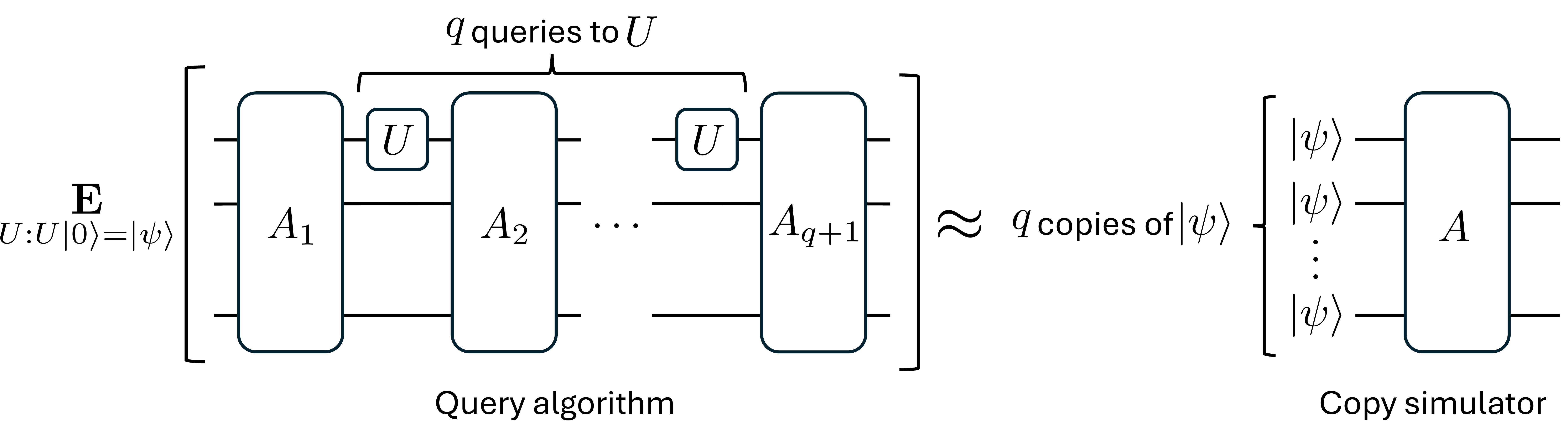}
    \caption{Illustration of Theorem~\ref{thm:main}. The left-hand algorithm makes $q$ forward queries to the same unitary $U$ satisfying $U\ket{0} = \ket{\psi}$. The unitaries $A_1,\ldots, A_{q+1}$ are applied before, between, and after these queries. The expectation is taken Haar-randomly over the remaining action of $U$. The right-hand simulator uses $q$ copies of $\ket{\psi}$, without access to $U$, to approximate this averaged output state. It is implemented by a unitary $A$ constructed from $A_1,\ldots, A_{q+1}$. The simulator is independent of the unknown state $\ket{\psi}$, and the $\approx$ placed in the center denotes closeness of the output states in Eq.~\eqref{eq: closeness-reduction}.}
    \label{fig:reductionmain}
\end{figure}

For the application to PAC learning, for any candidate query count $q$, we choose $N$ so that $q(r-1)/(N-1)\leq\delta/2$.
If a learner in this dimension has failure probability at most $\delta$ for every state-preparation unitary, its success guarantee is preserved when averaged over Haar-random unitaries.
The resulting $q$-copy learner obtained from the reduction theorem has failure probability at most $\delta+\delta/2$.
Since the simulator is independent of the unknown $\ket{\psi}$, a single copy learner satisfies this guarantee for all quantum examples used in the lower bound.
For $\delta\leq1/4$, this constant-factor change in the failure probability does not change the asymptotic dependence of the known copy lower bounds.
We can therefore transfer the copy lower bounds for realizable and agnostic learning to the respective forward-query lower bounds.

\paragraph{Proof overview of the reduction theorem.}
The key to the proof is to decompose the action of the state-preparation unitary into a part that is reproduced using copies and
a part whose contribution is bounded as an error.
Specifically, we decompose the action of $U$ into the action mapping the initial state $\ket{0}$ to $\ket{\psi}$,
the action whose output lies outside the known subspace $\cS$,
and the action with output in the subspace of $\cS$ orthogonal to $\ket{\psi}$.
Under Haar-random sampling, the second action can be chosen independently of the unknown $\ket{\psi}$.
Using this property and copies of $\ket{\psi}$,
we exactly reproduce the part of the Haar-averaged output to which only the first and second actions contribute.
The decomposition is detailed in Section~\ref{sec:explainprooftechnique}.

The input subspace corresponding to the remaining third action
is an $(r-1)$-dimensional subspace randomly embedded in an $(N-1)$-dimensional space.
Thus, fixing $r$ and increasing $N$
decreases the fraction of the dimension occupied by this subspace.
However, the algorithm can choose each input based on information obtained from earlier queries,
so this dimension ratio alone does not bound the overall error.
We show that, even allowing such adaptive operations,
the weight of the output to which the third action contributes is at most $q(r-1)/(N-1)$,
establishing the error guarantee of the reduction theorem.

\subsection{Related Works}

\paragraph{Forward-only learning at small ambient dimension.}
Chen~\cite{chen2025inversefreequantumstateestimation} presented a method for estimating the state $U\ket{0}$ using only forward queries to an unknown state-preparation unitary $U$ acting on an $N$-dimensional space.
For a fixed constant failure probability, the number of queries required to estimate the state to trace-distance error $\eta$ is
\begin{equation}
    O\!\left(
        \min\left\{
            \frac{N^{3/2}}{\eta},
            \frac{N}{\eta^2}
        \right\}
    \right).
\end{equation}
This result shows that even forward-only access can achieve better accuracy dependence than access to state copies alone, although its query complexity depends explicitly on the ambient dimension $N$.

Applying this state-estimation result to the PAC learning setting considered in this work yields forward-query upper bounds of
\begin{equation}
    O\!\left(\frac{N^{3/2}}{\sqrt{\varepsilon}}\right),
    \qquad
    O\!\left(\frac{N^{3/2}}{\varepsilon}\right)
\end{equation}
for realizable and agnostic learning, respectively, with constant failure probability.
In realizable learning, selecting the label with the larger amplitude magnitude in the estimated state yields a classification error of $O(\eta^2)$ from a trace-distance error of $\eta$.
In agnostic learning, choosing the hypothesis that minimizes the classification error with respect to the estimated state bounds the excess error relative to the true optimal error by $O(\eta)$.

In particular, when the ambient dimension on which the unitary actually acts satisfies $N=O(d)$, these bounds become
\begin{equation}
    O\!\left(\frac{d^{3/2}}{\sqrt{\varepsilon}}\right),
    \qquad
    O\!\left(\frac{d^{3/2}}{\varepsilon}\right).
\end{equation}
Thus, although the dependence on $d$ worsens, the dependence on $\varepsilon$ can improve quadratically over that achievable with classical samples or state copies.
This clarifies the importance of taking the worst case over ambient dimensions in our lower bounds.
In other words, our results do not assert that forward-only access offers no advantage for every fixed low-dimensional state-preparation unitary.

\paragraph{Agnostic learning with forward-and-inverse access.}
For realizable learning, it is shown that access to a state-preparation unitary $U$ and its inverse $U^\dagger$ enables PAC learning with
\begin{equation}
    O\!\left(
        \frac{d+\log(1/\delta)}{\sqrt{\varepsilon}}
        \log^9\frac{1}{\varepsilon}
    \right)
\end{equation}
queries~\cite{salmon2024provable}.
This improves the $1/\varepsilon$ accuracy dependence associated with classical or quantum examples to $1/\sqrt{\varepsilon}$, up to logarithmic factors, while maintaining a polynomial dependence on the VC dimension $d$.

For agnostic PAC learning, however, it remains unclear whether access
to both $U$ and $U^\dagger$ can generally improve the $1/\varepsilon^2$ accuracy dependence of classical samples and state copies while retaining a favorable dependence on the VC dimension $d$.
Our work shows that the $1/\varepsilon^2$ dependence is optimal in the worst case for forward-only access, but does not determine the optimal query complexity of agnostic learning when inverse access is also available.

State tomography also yields improved accuracy dependence when inverse access is available~\cite{vanapeldoorn2022quantumtomographyusingstatepreparation,tang2026controlledunitarieshelpful}~\footnote{The controlled queries used in the tomography algorithm can be removed without increasing the query count, since the task of estimation is invariant under the global phase of the unitary~\cite{tang2026controlledunitarieshelpful}.}.
The tomography results imply that an $N$-dimensional pure state can be estimated to trace-distance error $\eta$ using $\widetilde O(N/\eta)$ queries to its state-preparation unitary and its inverse, with constant failure probability.

Using the same conversion from state-estimation error to classification error as above gives query upper bounds
\begin{equation}
    \widetilde O\!\left(\frac{N}{\sqrt{\varepsilon}}\right),
    \qquad
    \widetilde O\!\left(\frac{N}{\varepsilon}\right)
\end{equation}
for realizable and agnostic learning, respectively.
In particular, when $N=O(d)$, these become
\begin{equation}
    \widetilde O\!\left(\frac{d}{\sqrt{\varepsilon}}\right),
    \qquad
    \widetilde O\!\left(\frac{d}{\varepsilon}\right).
\end{equation}
Here, $\widetilde O$ suppresses polylogarithmic factors.
These bounds improve the accuracy dependence while retaining a nearly linear dependence on $d$ in this low-dimensional setting.
However, the condition $N=O(d)$ is an additional restriction, i.e., it does not follow from the VC dimension alone.
Thus, these tomography bounds do not by themselves determine the optimal query complexity of general agnostic learning in terms of $d$, $\varepsilon$, and $\delta$ alone.

Our work shows that the $1/\varepsilon^2$ accuracy dependence is optimal for agnostic learning with forward-only access.
However, it does not establish a separation between forward-only access and forward-and-inverse access in agnostic learning.

\paragraph{Simulating state-preparation queries using copies.}
Tang, Wright, and Zhandry~\cite{tang2026conjugatequerieshelp} showed how to simulate queries to a state-preparation unitary using copies of the prepared state.
Specifically, under assumptions on the dimension of the ancillary system, the output of an algorithm making a total of $q$ queries to a state-preparation unitary $U$ and its inverse $U^\dagger$, averaged over a certain distribution of unitaries, can be reproduced using $O(q^2/\eta)$ copies to within trace distance $\eta$.

In contrast, our work considers forward-only access to a pure-state preparation unitary and the output averaged over Haar-random unitary completions satisfying $U\ket{0}=\ket{\psi}$.
When $\ket{\psi}$ belongs to a known $r$-dimensional subspace, the average output of any algorithm making $q$ forward queries can be approximated using the same number of copies, namely $q$, with trace-distance error at most
\begin{equation}
    \frac{q(r-1)}{N-1}.
\end{equation}
Thus, the two results allow different oracle access, and our result is not simply an improvement over the prior result in the same setting.
On the other hand, the $O(q^2/\eta)$ copy cost of the prior simulation does not directly transfer known copy lower bounds to query lower bounds of the same order.
By restricting attention to forward-only access, we construct a reduction that uses as many copies as queries and controls the approximation error through the ambient dimension $N$.
This allows known copy lower bounds to be transferred to forward-query lower bounds without asymptotic loss.

\subsection{Conclusion}
We have characterized the query complexity of quantum PAC learning with forward-only access to state-preparation unitaries.
Requiring success for every state-preparation unitary that generates the same quantum example, we established query lower bounds for realizable and agnostic learning in a worst-case setting with unrestricted ambient dimension.
These lower bounds match, up to constant factors, the upper bounds achieved by measuring quantum examples and applying an optimal classical learner to the resulting classical examples.
Our lower bounds are therefore asymptotically optimal and determine the query complexity in both learning settings.
The proof constructs a reduction that approximates the Haar-averaged output of a forward-query algorithm using state copies, and applies known lower bounds for copies of quantum examples.

These results clarify the learning power of forward-only access in relation to the model that supplies copies of quantum examples and the model that permits both a state-preparation unitary and its inverse.
In particular, together with the known improvement in accuracy dependence from forward-and-inverse access in realizable learning, they show that the ability to apply a state-preparation unitary to arbitrary inputs alone does not achieve the same improvement in our worst-case setting, and that inverse access is crucial to achieve the quadratic improvement.

Furthermore, the reduction theorem itself does not assume PAC learning and is applicable to any tasks with access to an unknown state-preparation unitary beyond the PAC learning.
The benefit of this reduction is that it allows the limitations of algorithms using forward access to be assessed through the limitations on information obtainable from state copies.
Unitary access permits arbitrary superpositions of inputs and adaptive operations based on information from previous queries, making its power difficult to analyze directly.
In settings where the approximation error can be made sufficiently small, our reduction allows known lower bounds for state copies to be used in place of a task-specific analysis of these operations.
Our results therefore provide a tool to investigate whether forward access to state-preparation unitaries can reduce the number of queries needed for tasks with known copy lower bounds.

\paragraph{Organization.}
This paper is organized as follows.
Section~\ref{sec:preliminaries} introduces quantum information notation and the access models for PAC learning.
Section~\ref{sec:reduction-overview} proves Theorem~\ref{thm:copy-reduction}.
Section~\ref{sec:pac-consequences} applies this reduction to realizable and agnostic learning to establish the lower bounds in Table~\ref{tab:pac-oracle-bounds}.
Some technical proofs are deferred to the appendices.

\section{Preliminaries}\label{sec:preliminaries}
In this section, we specify the quantum information notation and learning models used in this paper and review relevant known results.
First, Section~\ref{sec:basics} introduces basic notation for quantum states and quantum operations and defines the trace distance, which measures the closeness of states.
In Section~\ref{sec:pac-learning}, starting from the classical PAC learning model, we introduce a model that receives copies of quantum examples and a model with access to a state-preparation unitary.
For each model, we define realizable and agnostic learning and specify the operations available to the learner and the success conditions it must satisfy.
We then define the Vapnik-Chervonenkis (VC) dimension, an important parameter for measuring sample complexity, and review the known sample and query complexities in each model.

\subsection{Basic notation and distances}\label{sec:basics}
For an integer $m\geq1$, let $[m]=\{1,\ldots,m\}$.
We write $|A|$ for the cardinality of a set $A$.
$\E_{x\sim\mathcal D}[f(x)]$ denotes the expectation of the function $f(x)$ of a random variable $x$ distributed according to the probability distribution $\mathcal D$.
$\Prb_{x\sim\mathcal D}[E(x)]$ denotes the probability that a random variable $x$ distributed according to the probability distribution $\mathcal D$ satisfies the event $E(x)$.

All Hilbert spaces considered in this work are finite-dimensional complex Hilbert spaces.
For a subspace $\cS\subseteq\cH$, we write $P_{\cS}$ for the orthogonal projector onto $\cS$.
For a unit vector $\ket{v}$, we write $P_v=\proj{v}$ and $\ket{v}^{\perp}\coloneq \{\ket{w}\in\cH:\braket{v|w}=0\}$.
For each Hilbert space, we fix a computational basis with respect to which transposition and entrywise complex conjugation are taken.
For an ordinary linear operator $A$, the symbols $A^\top$, $A^\ast$, and $A^\dagger=(A^\ast)^\top$ denote its transpose, entrywise complex conjugate, and adjoint, respectively, in the fixed basis.

Let $\mathcal{I}$ and $\mathcal{O}$ be finite-dimensional Hilbert spaces representing the input and output spaces of a quantum operation.
We denote the space of linear operators from $\mathcal{I}$ to $\mathcal{O}$ by $\mathsf{L}(\mathcal{I},\mathcal{O})$ and abbreviate $\mathsf{L}(\mathcal{I},\mathcal{I})$ as $\mathsf{L}(\mathcal{I})$.
A quantum state on $\mathcal{I}$ is an operator $\rho \in \mathsf{L}(\mathcal{I})$ satisfying
\begin{equation}
\rho\succeq0, \qquad \Tr(\rho)=1.
\end{equation}
An operator $\rho\succeq0$ satisfying $\Tr(\rho)\leq1$ is called a subnormalized state.

A linear map $\Phi:\mathsf{L}(\mathcal{I})\to\mathsf{L}(\mathcal{O})$ is completely positive if, for every finite-dimensional Hilbert space $\mathcal{R}$, the map $\id_{\mathcal{R}}\otimes\Phi$ maps positive semidefinite operators to positive semidefinite operators.
A map is trace-nonincreasing if, for every positive semidefinite operator $X\in\mathsf{L}(\mathcal{I})$, it satisfies $\Tr[\Phi(X)]\leq\Tr[X]$.
A map for which this inequality always holds with equality is called trace-preserving, and a completely positive and trace-preserving map is called a quantum channel.
The trace distance between density operators $\rho$ and $\sigma$ is defined by
\begin{equation}
    \Dtr(\rho,\sigma)=\frac{1}{2}\norm{\rho-\sigma}_1,
\end{equation}
where $\norm{X}_1=\Tr\sqrt{X^\dagger X}$ is the 1-norm.

\subsection{Probably Approximately Correct (PAC) learning models}\label{sec:pac-learning}
In this section, following~\cite{arunachalam2017survey,arunachalam2018optimal,salmon2024provable}, we introduce the definitions of classical Probably Approximately Correct (PAC) learning and quantum PAC learning with quantum examples.
We introduce an access model for a state-preparation unitary that generates quantum examples.
We further extend these models to agnostic learning.
We review known results for these access models.

\subsubsection{Classical PAC model}
\paragraph{Classical realizable learning.}
Let $\cX$ be a finite input space and $\cY=\{0,1\}$ be the label space.
We denote by $\{0,1\}^\cX = \{f:\cX\to\{0,1\}\}$ the set of all functions from $\cX$ to $\{0,1\}$.
A subset $\cC\subseteq\{0,1\}^{\cX}$ is called a concept class, and a function $c \in \cC$ is called a concept.
For an unknown concept $c\in\cC$ and an unknown distribution $\mathcal D$ over the input space, we define the random example oracle $\mathrm{PEX}(c,\mathcal{D})$.
When queried, this oracle samples $x\sim\mathcal{D}$ and returns an input--output pair $(x,c(x))$.
The pair $(x,c(x))$ is also called an example.

In the PAC model, a learner $\mathcal A$ is given access to $\mathrm{PEX}(c,\mathcal{D})$ for an unknown concept $c\in\cC$ and an unknown distribution $\mathcal D$, and aims to output a hypothesis $h:\cX\to\cY$ that approximates the concept $c$.
Although the original definition of PAC learning requires efficient computation~\cite{valiant1984theory,kearns1990computational, kearns1994introduction}, this paper considers the required number of examples or oracle queries rather than computation time.
In this paper, we follow~\cite{arunachalam2017survey,arunachalam2018optimal} in defining the success conditions for learning algorithms.
\begin{definition}[Realizable $(\varepsilon,\delta)$-PAC learner]\label{def:realizable-pac-learner}
Let $0<\varepsilon,\delta<1$. Define the error of a hypothesis $h:\cX\to\cY$ with respect to a concept $c\in\cC$ and an input distribution $\mathcal D$ by
\begin{equation}
    \err_{\mathcal D,c}(h)=\Prb_{x\sim\mathcal D}[h(x)\neq c(x)]
\end{equation}
We say that an algorithm $\mathcal A$ is an $(\varepsilon,\delta)$-PAC learner for the concept class $\cC$ if, for every $c\in\cC$ and every distribution $\mathcal D$ over $\cX$, given access to $\mathrm{PEX}(c,\mathcal D)$, the algorithm $\mathcal A$ outputs a hypothesis $h$ such that
\begin{equation}
    \Prb\!\left[\err_{\mathcal D,c}(h)\leq\varepsilon\right]\geq1-\delta,
\end{equation}
where the outer probability is taken over the examples returned by the oracle and the internal randomness of the learning algorithm.
\end{definition}
Note that this definition does not require the learner to output a hypothesis $h\in\cC$ that exactly represents a concept in $\cC$.
The case in which $h\in\cC$ is required is called proper, and the case in which it is not required is called improper.

\paragraph{Classical agnostic learning.}
In standard (realizable) PAC learning, the labels of the examples received are assumed to follow some concept $c\in\cC$, i.e., $y=c(x)$ holds.
In agnostic learning, by contrast, the labels are not assumed to follow a concept.
For an unknown probability distribution $\cD$ over $\mathcal{X} \times \mathcal{Y}$, the learner receives access to $\mathrm{AEX}(\cD)$ rather than $\mathrm{PEX}(c,\mathcal{D})$.
When queried, $\mathrm{AEX}(\cD)$ returns a labeled example with $(x,y)\sim\cD$.

In this setting, there need not exist a concept or hypothesis that perfectly represents the unknown distribution $\cD$.
The error of a hypothesis $h$ is defined by
\begin{equation}
    \err_{\cD}(h)=\Prb_{(x,y)\sim\cD}[h(x)\neq y],
\end{equation}
and the minimum error achievable within $\cC$ is defined by
\begin{equation}
    \opt_{\cC}(\cD)=\inf_{c\in\cC}\err_{\cD}(c).
\end{equation}
The learner aims to learn a hypothesis $h$ from the given examples whose error is close to the minimum error achievable within the given concept class $\cC$.
Specifically, the learning condition is defined as follows~\cite{10.1145/130385.130424}.

\begin{definition}[$(\varepsilon,\delta)$-agnostic PAC learner]\label{def:agnostic-pac-learner}
Let $0<\varepsilon,\delta<1$.
We say that an algorithm $\mathcal A$ is an $(\varepsilon,\delta)$-agnostic PAC learner for the concept class $\cC$ if, for every distribution $\cD$ over $\cX\times\cY$, given access to $\mathrm{AEX}(\cD)$, the algorithm $\mathcal A$ outputs a hypothesis $h:\cX\to\cY$ such that
\begin{equation}
    \Prb\!\left[\err_{\cD}(h)\leq\opt_{\cC}(\cD)+\varepsilon\right]\geq1-\delta
\end{equation}
where the outer probability is taken over the examples returned by the oracle and the internal randomness of the learning algorithm.
\end{definition}

In this paper, we also allow improper outputs and consider the required number of examples or oracle queries rather than computation time.
A realizable pair $(c,\mathcal D)$ induces the joint distribution $\mathcal{D}_c$ given by
\begin{equation}
    \mathcal{D}_c(x,y) = \begin{cases}
        \mathcal{D}(x), &\quad y = c(x),\\
        0, &\quad y \neq c(x).
    \end{cases}
\end{equation}
In this case, $\opt_{\cC}(\cD_c)=0$, and the realizable and agnostic success conditions coincide for this induced distribution.

\subsubsection{Quantum PAC model with coherent quantum examples}
\paragraph{Quantum example oracles.}
We consider a model in which the learner receives quantum states encoding example probabilities as squared amplitudes in place of classical examples.
We represent the example labels by a known orthonormal family $\{\ket{x,y}:x\in\cX,\ y\in\cY\}$.
In the realizable setting, the quantum example oracle $\mathrm{QPEX}(c,\mathcal{D})$ for an input distribution $\mathcal D$ and a concept $c\in\cC$ generates one copy of the following state from the known initial state $\ket{0}$~\cite{arunachalam2018optimal}.
\begin{equation}\label{eq:realizable-qexample}
    \mathrm{QPEX}(c,\mathcal{D}):\ket{0}\longmapsto\ket{\psi_{\mathcal{D},c}}
    =\sum_{x\in\cX}\sqrt{\mathcal{D}(x)}\ket{x,c(x)}.
\end{equation}
Similarly, for a joint distribution $\cD$ over $\cX\times\cY$, the quantum agnostic example oracle $\mathrm{QAEX}(\cD)$ generates the following state~\cite{arunachalam2018optimal}.
\begin{equation}\label{eq:agnostic-qexample}
    \mathrm{QAEX}(\cD):\ket{0}\longmapsto\ket{\psi_{\cD}}
    =\sum_{x\in\cX}\sum_{y\in\cY}\sqrt{\cD(x,y)}\ket{x,y}.
\end{equation}
In either case, each query provides one fresh copy of a quantum example. Measuring in the computational basis yields one classical example returned by $\mathrm{PEX}(c,\mathcal{D})$ or $\mathrm{AEX}(\cD)$, respectively. The oracles $\mathrm{QPEX}$ and $\mathrm{QAEX}$ specify only the state generated from the initial state and do not provide the learner with the unitary's action on other inputs.

\paragraph{Learners using quantum examples.}
The learner may retain the copies received and perform joint quantum operations on the copies, intermediate measurements, and classical control before outputting a classical hypothesis; improper outputs are also allowed.
The success conditions are the same as in the classical model; only the input resources differ.

\begin{definition}[Realizable $(\varepsilon,\delta)$-PAC learner using quantum examples]\label{def:qpex-pac-learner}
Let $0<\varepsilon,\delta<1$.
We say that a quantum algorithm $\mathcal A$ is a realizable $(\varepsilon,\delta)$-PAC learner for the concept class $\cC$ using quantum examples if, for every $c\in\cC$ and every distribution $\mathcal D$ over $\cX$, given access to $\mathrm{QPEX}(c,\mathcal{D})$, the algorithm $\mathcal A$ outputs a hypothesis $h:\cX\to\cY$ such that
\begin{equation}
    \Prb\!\left[\err_{\mathcal{D},c}(h)\leq\varepsilon\right]\geq1-\delta
\end{equation}
where the outer probability is taken over the measurement outcomes and the internal randomness of the learning algorithm.
\end{definition}
\begin{definition}[$(\varepsilon,\delta)$-agnostic PAC learner using quantum examples]\label{def:qaex-pac-learner}
Let $0<\varepsilon,\delta<1$.
We say that a quantum algorithm $\mathcal A$ is an $(\varepsilon,\delta)$-agnostic PAC learner for the concept class $\cC$ using quantum examples if, for every distribution $\cD$ over $\cX\times\cY$, given access to $\mathrm{QAEX}(\cD)$, the algorithm $\mathcal A$ outputs a hypothesis $h:\cX\to\cY$ such that
\begin{equation}
    \Prb\!\left[\err_{\cD}(h)\leq\opt_{\cC}(\cD)+\varepsilon\right]\geq1-\delta
\end{equation}
where the outer probability is taken over the measurement outcomes and the internal randomness of the learning algorithm.
\end{definition}

\subsubsection{Quantum PAC model with state-preparation-unitary access}
The model in which the learner only receives quantum examples differs from the model in which it can query a unitary operation that generates those examples.
In the latter model, a unitary $U$ that maps a known initial state $\ket{0}$ to a quantum example $\ket{\psi}$ can be applied to arbitrary query states during learning.
A unitary $U$ generating the same $\ket{\psi}$ is not unique, and its action on inputs other than $\ket{0}$ is not uniquely determined by the quantum example alone.
\begin{definition}[Access to a state-preparation unitary]\label{def:state-preparation-access}
For a quantum example $\ket{\psi}$, suppose that the learner can repeatedly query the same unknown unitary $U$ satisfying $U\ket{0}=\ket{\psi}$ for a known initial state.
Access to only $U$ is called forward-only access, whereas access to both $U$ and $U^\dagger$ is called forward-and-inverse access.
Each application is counted as one query.
\end{definition}
In the forward-only model, each oracle query applies $U$
unconditionally.
Access to $U^\dagger$, $U^\top$, $U^\ast$, and controlled-$U$ is not supplied as an additional primitive. 
Arbitrary known channels, ancillary systems, intermediate measurements, and classical control are allowed between oracle queries. Operations synthesized using forward queries are allowed, with all such queries included in the query count.

For every permitted ambient dimension $N$, a forward-only learner must satisfy the success condition in Definition~\ref{def:realizable-pac-learner} for every realizable instance $(c,D)$ and every unitary satisfying $U\ket{0}=\ket{\psi_{D,c}}$.
In the agnostic setting, it must satisfy the success condition in Definition~\ref{def:agnostic-pac-learner} for every joint distribution $\cD$ and every unitary satisfying $U\ket{0}=\ket{\psi_{\cD}}$.

With forward-and-inverse access, $U^\dagger$ can be used to implement a reflection about the quantum example.
\begin{equation}
    I-2 \proj{\psi}=U\bigl(I-2\ketbra{0}{0}\bigr)U^\dagger.
\end{equation}
This reflection can be used in amplitude amplification.
Salmon, Strelchuk, and Gur treat this forward-and-inverse access as a state-preparation oracle in realizable PAC learning and show an advantage over classical access~\cite{salmon2024provable} (see also Section~\ref{sec:samplecomplexity summerize}).

\subsubsection{Sample complexity and the number of state-preparation queries}\label{sec:samplecomplexity summerize}
In the classical example model, sample complexity is measured by the worst-case number of queries to $\mathrm{PEX}$ or $\mathrm{AEX}$; in the quantum example model, it is measured by the number of copies received from $\mathrm{QPEX}$ or $\mathrm{QAEX}$.
In the state-preparation-unitary model, we count queries to $U$ and, when available, $U^\dagger$.
Below, computation time is not included in the complexity.

For forward-only and forward-and-inverse PAC learning, the learner may depend on the known ambient dimension $N$ and the known encoding.
It may therefore be specified by a family of algorithms $(\mathcal A_N)_{N \in \mathbb{Z}_{\geq 2}}$.
Its worst-case query count $q$ is taken over all permitted ambient dimensions, all learning instances, all compatible unitary completions, and all branches of the algorithm.
Thus, a finite bound $q$ must hold uniformly in $N$, although the algorithm itself may vary with $N$.
The optimal query complexities for these models are obtained over learner families.

The PAC lower bounds in Section~\ref{sec:pac-consequences}
use this convention. 
By contrast, the reduction theorem in Section~\ref{sec:reduction-overview} retains an explicit dependence on the fixed ambient dimension $N$.

In the PAC model, sample complexity depends on the Vapnik--Chervonenkis (VC) dimension, a parameter determined by the concept class $\cC$.
Let $S=\{x_1,\ldots,x_m\}\subseteq\cX$ be a subset of the input space.
The restriction of a function $f: \mathcal{X} \to \{0,1\}$ to $S$ is denoted by $f|_S : S \to \{0,1\}$.
We also define the restriction of the concept class $\cC$ to $S$ by $\cC|_S=\{c|_{S}:c\in\cC\}$.
We say that $\cC$ shatters $S$ if $\cC|_S = \{0,1\}^S$.
Using this definition, we define the VC dimension as follows.

\begin{definition}[VC dimension]\label{def:vc-dimension}
Let $S=\{x_1,\ldots,x_m\}\subseteq\cX$ be a subset of the input space.
The VC dimension $\vc(\cC)$ of $\cC$ is the maximum cardinality of a set shattered by $\cC$, i.e.,
\begin{equation}
    \vc(\cC)=\max\{|S|:S\subseteq\cX,\ \cC|_S=\{0,1\}^S\}.
\end{equation}
\end{definition}

Unless otherwise specified, let $d=\vc(\cC)$ denote the VC dimension of $\cC$ below.
For realizable learning allowing improper hypotheses, the optimal sample complexities $q$ of a classical $(\varepsilon,\delta)$-PAC learner using $\mathrm{PEX}(c,\mathcal{D})$ and a quantum $(\varepsilon,\delta)$-PAC learner receiving copies of quantum examples from $\mathrm{QPEX}(c,\mathcal{D})$ are the same, i.e., for $d \geq 2$, $\varepsilon \in (0,1/32)$, and $\delta \in (0,1/4)$,
\begin{equation}
    q = \Theta\!\left(\frac{d+\log(1/\delta)}{\varepsilon}\right)
\end{equation}
as shown in~\cite{10.1145/76359.76371,hanneke2016optimal,EHRENFEUCHT1989247,arunachalam2018optimal}.
In agnostic learning as well, the optimal sample complexities of a classical $(\varepsilon,\delta)$-agnostic PAC learner using $\mathrm{AEX}(\cD)$ and a quantum $(\varepsilon,\delta)$-agnostic PAC learner receiving copies of quantum examples from $\mathrm{QAEX}(\cD)$ are also the same, i.e., for $d \geq 1$, $\varepsilon \in (0,1/32)$, and $\delta \in (0,1/4)$,
\begin{equation}
    q = \Theta\!\left(\frac{d+\log(1/\delta)}{\varepsilon^2}\right)
\end{equation}
as shown in~\cite{vapnik74theory,a0b266f6-398c-301c-9bea-dc3d41c6ccf7,10.1145/168304.168385,arunachalam2018optimal}.

For realizable learning with forward-and-inverse access, Salmon, Strelchuk, and Gur showed that learning is possible with
\begin{equation}
    O\!\left(\frac{d+\log(1/\delta)}{\sqrt{\varepsilon}}\log^9\frac1\varepsilon\right)
\end{equation}
queries to $U$ or $U^\dagger$~\cite{salmon2024provable}.
Compared with copies, this gives a quadratic improvement in the dependence on $\varepsilon$.
For the forward-only access model, however, it remained open whether such an improvement in sample complexity over classical examples $\mathrm{PEX},\mathrm{AEX}$ or quantum examples $\mathrm{QPEX},\mathrm{QAEX}$ could be achieved.

In Section~\ref{sec:pac-consequences}, we show that the forward-only access model offers no such advantage over the sample complexities of classical examples and quantum example copies.
Theorem~\ref{thm:realizable-main} and Theorem~\ref{thm:agnostic-main} of this paper transfer sample complexity lower bounds for quantum example copies to the forward-only access model through a reduction, showing that both realizable and agnostic learning have the same asymptotic complexity as learning from quantum example copies, even for improper learners.
Since measuring a quantum example yields a classical example, the classical upper bounds are also achievable in the quantum example model.
For the lower bounds, we use the following known result.

\begin{lemma}[Lemmas~11 and~12 and Theorems~23 and~25 of~\cite{arunachalam2018optimal}]\label{lem:qsample-copy-bound}
Let $d=\vc(\cC)$ be the VC dimension.
The notation $\Omega$ below denotes an asymptotic bound for each fixed $0<\delta<1/2$.
For any fixed $\delta_0<1/2$, the implicit constants can be chosen independently of $\delta$ throughout the range $0<\delta\leq\delta_0$.
\begin{enumerate}[label=(\roman*)]
    \item For $d\geq2$, $0<\varepsilon<1/20$, and $0<\delta<1/2$, the sample complexity $q$ of any $(\varepsilon,\delta)$-PAC learner with access to $\mathrm{QPEX}(c,\mathcal{D})$ as in Definition~\ref{def:qpex-pac-learner} (allowing improper outputs) satisfies
    \begin{equation}
        q=\Omega\!\left(\frac{d+\log(1/\delta)}{\varepsilon}\right).
    \end{equation}
    \item For $d\geq1$, $0<\varepsilon<1/10$, and $0<\delta<1/2$, the sample complexity $q$ of any $(\varepsilon,\delta)$-agnostic PAC learner with access to $\mathrm{QAEX}(\cD)$ as in Definition~\ref{def:qaex-pac-learner} (allowing improper outputs) satisfies
    \begin{equation}
        q=\Omega\!\left(\frac{d+\log(1/\delta)}{\varepsilon^2}\right).
    \end{equation}
\end{enumerate}
\end{lemma}
These lower bounds already hold when the learner is required to succeed only on distributions whose inputs are supported on a fixed set of $d$ points shattered by $\mathcal{C}$~\cite{arunachalam2017survey,arunachalam2018optimal}.
Although Theorem 23 in the original paper in~\cite{arunachalam2018optimal} is stated with VC dimension $d+1$, here, we have adapted it to this paper's asymptotically equivalent notation $d=\vc(\cC)$.
The proofs of both theorems establish the lower bounds using only families of distributions supported on a fixed finite input set, so the same distribution families can be used in the reduction in the following sections.
Since values of $d$ that are not sufficiently large are bounded, the confidence lower bound obtained from two-state discrimination also absorbs the $d$ term.
The lower bounds in Theorems~23 and~25 of~\cite{arunachalam2018optimal} are proved only for sufficiently large $d$.
For smaller values of $d$, which are bounded by an absolute constant, the confidence-dependent lower bounds in Lemmas~11 and~12 of~\cite{arunachalam2018optimal} still imply the desired lower bounds in Lemma~\ref{lem:qsample-copy-bound}, since $\log(1/\delta) = \Omega(1)$ and hence $d + \log(1/\delta) = O(\log(1/\delta))$.

The agnostic learner in the original paper outputs $h\in\cC$. However, on the support of the distributions used in its lower bound, there is a concept $c_h\in\cC$ that assigns the same labels as any given hypothesis $h$. Replacing $h$ with $c_h$ leaves the error on these distributions unchanged and does not increase the number of copies. Thus, the agnostic lower bound also holds when improper outputs are allowed.

\section{Reduction from forward state-preparation unitaries to quantum states}\label{sec:reduction-overview}
Section~\ref{sec:pac-learning} defined PAC learners given only forward access to a unitary that generates quantum examples.
A state-preparation unitary generates an unknown quantum example $\ket{\psi}$ from a known initial state $\ket{0}$, leaving unitary freedom in the remaining columns.
In this paper, we require success for every state-preparation unitary satisfying $U\ket{0}=\ket{\psi}$.
The success guarantee therefore remains valid when the remaining columns are chosen Haar-randomly.
In this section, we show that the output of a quantum algorithm making forward-only queries to a state-preparation unitary, averaged over the Haar distribution of the unitary completion, can be generated to a certain accuracy by a quantum algorithm receiving only copies of the quantum state.
Section~\ref{sec:forward-model} introduces the forward-only state-preparation-unitary model, formulates the problem, and states the theorem.
Section~\ref{sec:explainprooftechnique} explains the key ideas of the proof.
Section~\ref{sec:proof2} proves the theorem.
Using this theorem, Section~\ref{sec:pac-consequences} establishes query complexity lower bounds for PAC learners in the forward-only access model.

\subsection{The forward state-preparation access setting and the reduction theorem}\label{sec:forward-model}
We first specify the notation used in this section.
Let $\cH\simeq\C^N$ be a finite-dimensional Hilbert space with a known basis $\{\ket{0},\ldots,\ket{N-1}\}$.
Fix a known initial state $\ket{0}\in\cH$.
For simplicity of notation, we write
\begin{equation}
    P_0=\ketbra{0}{0},
    \qquad
    Q_0=I_\cH - P_0,
    \qquad
    \ket{0}^{\perp} \coloneqq \operatorname{span}\{\ket{1},\ldots,\ket{N-1}\},
    \qquad
    \dim \ket{0}^{\perp} =N-1.
\end{equation}
Let $r=\dim\cS$ be the dimension of a known subspace $\cS\subseteq\ket{0}^{\perp}$.

We assume that the unknown quantum state $\ket{\psi}$ belongs to the known subspace $\cS$.
A unitary $U$ generating $\ket{\psi}\in\cS$ specifies a map from the known initial state $\ket{0}$ to $\ket{\psi}$, but its action on $\ket{0}^{\perp}$ is not uniquely determined by $\ket{\psi}$.
In particular, the remaining columns of $U$ can be chosen as any orthonormal basis of $\ket{\psi}^\perp$.
This setting is illustrated in Fig~\ref{fig:Hilbertspacesetting}.

\begin{figure}
    \centering
    \includegraphics[width=0.6\linewidth]{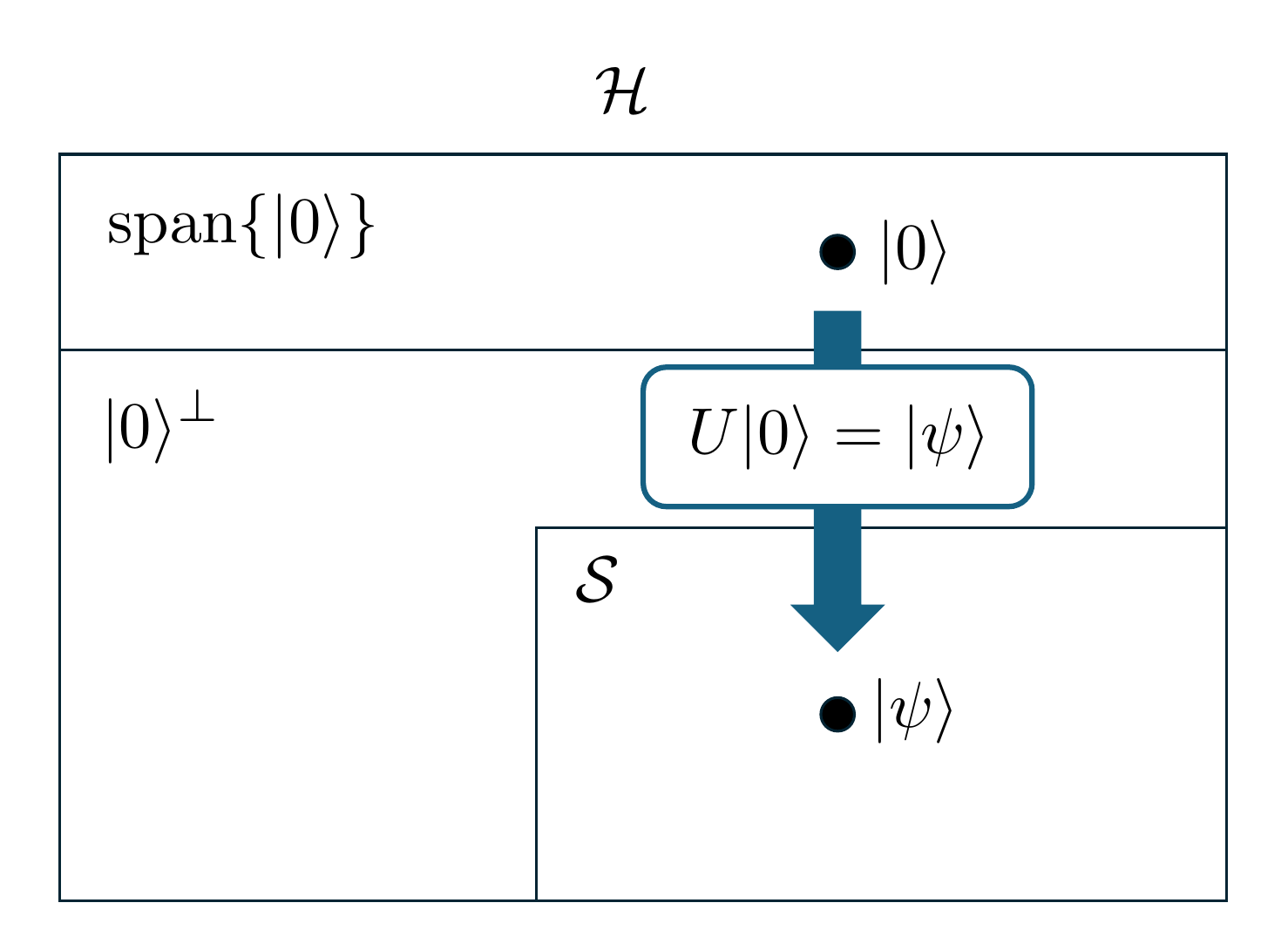}
    \caption{Illustration of the state-preparation setting. The unknown state $\ket{\psi}$ lies in a known subspace $\cS \subseteq \ket{0}^\perp$ in $\cH$. A unitary $U$ maps the known initial state $\ket{0}$ to $\ket{\psi}$.}
    \label{fig:Hilbertspacesetting}
\end{figure}

\begin{definition}[State-preparation unitary]\label{def:completion}
A state-preparation unitary for a unit vector $\ket{\psi}\in\cS$ is a unitary $U$ satisfying
\begin{equation}
    U\ket{0}=\ket{\psi}.
\end{equation}
Equivalently, it can be written as
\begin{equation}\label{eq:completion-decomposition}
    U=\ketbra{\psi}{0}+WQ_0,
\end{equation}
where $W: \ket{0}^{\perp} \to\ket{\psi}^{\perp}$ is a unitary isomorphism between two spaces of dimension $N - 1$.
\end{definition}

A $q$-query forward-only algorithm may insert arbitrary known channels, ancillary systems, measurements, and classical control between the $q$ applications of the same unknown $U$.
However, it is not given $U^\dagger$, $U^\top$, $U^\ast$, or controlled $U$.
By adding queries on initial-state registers that are subsequently discarded to branches that have halted, any algorithm using at most $q$ queries can be regarded as making exactly $q$ unconditional queries.

For particular choices, $U=U^\dagger$, so the inverse operation can be performed using forward queries alone.
The remaining columns may also contain additional information about $\ket{\psi}$, and the algorithm's behavior may depend on the choice of completion.
However, the reduction theorem below does not concern the output for a particular choice of $U$.
With $\ket{\psi}$ fixed, we choose the free part $W$ in Eq.~\eqref{eq:completion-decomposition} uniformly at random and consider the algorithm's output averaged over this distribution, as defined below.

\begin{definition}[Haar distribution of state-preparation unitaries]\label{def:haar-completion}
For $\ket{\psi}\in\cS$, let $\HC_N(\psi)$ be the distribution of $U$ obtained by choosing the unitary isomorphism $W:\ket{0}^{\perp}\to\ket{\psi}^{\perp}$ in Eq.~\eqref{eq:completion-decomposition} from the uniform Haar-invariant distribution.
\end{definition}

\begin{theorem}[Copy simulation of the Haar-averaged output]\label{thm:copy-reduction}
Let $\cS\subseteq\ket{0}^{\perp}$ be a known subspace of dimension $r$, and let $\mathcal{A}_q$ be any quantum algorithm making $q$ forward queries to a state-preparation unitary $U$ as in Definition~\ref{def:completion}.
For each $U$, let $\rho_{\mathcal A_q}(U)$ be the final output state of $\mathcal A_q$ when the same $U$ is used in every query.
Suppose that
\begin{equation}
    q(r-1)<N-1.
\end{equation}
Then there exists a quantum channel $\mathfrak S_q$ independent of the choice of $\ket{\psi} \in \cS$.
For this channel, the following holds for every unit vector $\ket{\psi}\in\cS$.
\begin{equation}\label{eq:copy-reduction-state}
    \Dtr\!\left(
      \E_{U\sim\HC_N(\psi)}[\rho_{\mathcal{A}_q}(U)],
      \mathfrak S_q(P_\psi^{\otimes q})
    \right)
    \leq \frac{q(r-1)}{N-1}.
\end{equation}
\end{theorem}

In the following, we prove Theorem~\ref{thm:copy-reduction}.

\subsection{Decomposition of the output quantum state and simulation using state copies}~\label{sec:explainprooftechnique}
In this section, we introduce the main ideas in the proof of Theorem~\ref{thm:copy-reduction}.
First, we decompose the action of a state-preparation unitary into three blocks to explicitly separate the parts that depend on $\ket{\psi}$.
Using this decomposition, we split the final output state of the forward-only algorithm $\mathcal{A}_q$ into two parts.
One part can be generated from copies of the quantum state.
The other part depends on $\ket{\psi}$ and is therefore bounded as an error term.

Fix an unknown unit vector $\ket{\psi}\in\cS$.
Within the known space $\cS$, define the following subspace orthogonal to $\ket{\psi}$:
\begin{equation}
    \cB_\psi=\cS\cap\ket{\psi}^{\perp}.
\end{equation}
Then the following holds.
\begin{equation}
    \ket{\psi}^{\perp}=\cS^{\perp}\oplus\cB_\psi,
    \qquad
    \dim\cB_\psi=r-1.
\end{equation}
Consider the subspace of $\ket{0}^{\perp}$ that $U$ maps to $\cB_\psi$.
In terms of $W$ from Eq.~\eqref{eq:completion-decomposition}, this subspace is $W^{-1}(\cB_\psi)$.
Let $\Pi$ be the orthogonal projector onto this subspace.
Since $\dim\cB_\psi=r-1$, the rank of $\Pi$ is also $r-1$.
\begin{restatable}[Decomposition of a state-preparation unitary]{lemma}{haardecomposition}
\label{lem:haar-block-decomposition}
A Haar-random state-preparation unitary for $\ket{\psi}$ can be generated by first choosing a Haar-random projector $\Pi$ of rank $(r-1)$ on $\ket{0}^{\perp}$ and then, conditioned on $\Pi$, choosing the following independent Haar-random isometries,
\begin{equation}
    G_\Pi:(I-\Pi)\ket{0}^{\perp}\longrightarrow\cS^{\perp},
    \qquad
    B_{\Pi,\psi}:\Pi\ket{0}^{\perp}\longrightarrow\cB_\psi.
\end{equation}
When these maps are extended by zero outside their domains and independent phases $\theta,\phi\in[0,2\pi)$ are chosen, the following unitary
\begin{equation}\label{eq:note-decomposition}
    U_{\theta,\phi,\Pi,G,B}
    =\ketbra{\psi}{0}+e^{i\theta}G_\Pi+e^{i\phi}B_{\Pi,\psi}
\end{equation}
follows the Haar distribution.
In particular, $\Pi$ and $G_\Pi$ do not depend on $\ket{\psi}$.
\end{restatable}
\begin{figure}
    \centering
    \includegraphics[width=0.65\linewidth]{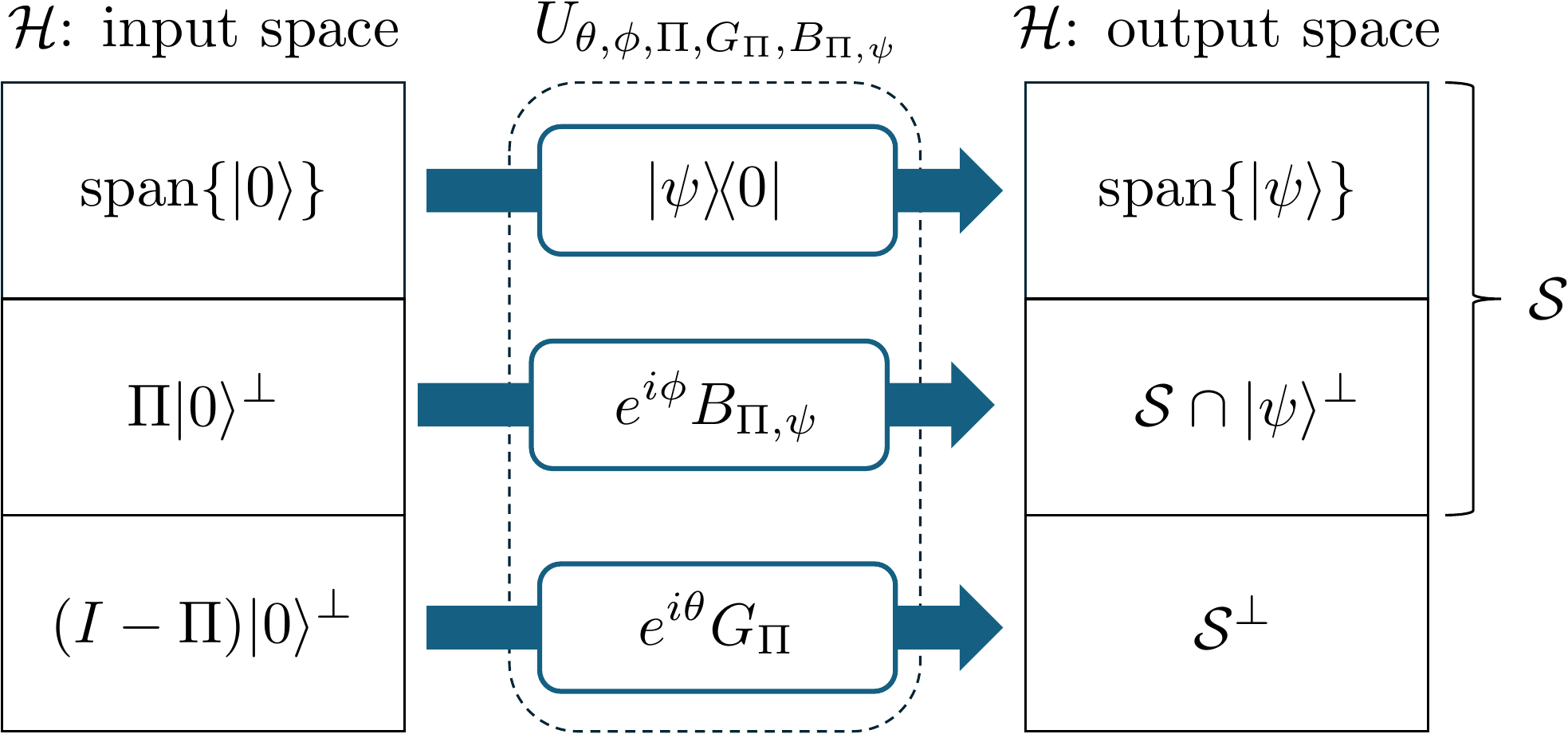}
    \caption{Illustration of the unitary decomposition. Both the input and output spaces are $\cH$, but they are decomposed into different orthogonal subspaces. The three blocks of the unitary $U_{\theta,\phi,\Pi,G,B}$ map the input subspaces to the corresponding output subspaces.}
    \label{fig:unitary-decomposition}
\end{figure}

We illustrate this decomposition in Fig.~\ref{fig:unitary-decomposition}.
The proof is given in Appendix~\ref{app:haar-decomposition}.
By Lemma~\ref{lem:haar-block-decomposition}, $\Pi$ and $G_\Pi$ can be chosen independently of $\ket{\psi}$.
Although $\ketbra{\psi}{0}$ depends on $\ket{\psi}$, its contribution can be reproduced using copies of $\ket{\psi}$, as we show later.
In contrast, $B_{\Pi,\psi}$ depends on $\ket{\psi}$.
Below, we call the three terms in Eq.~\eqref{eq:note-decomposition} the main, good, and bad blocks, respectively.
We reproduce components containing only main and good blocks from copies and bound the weight of components containing bad blocks as an error term.

We substitute this three-block decomposition into the $q$ queries to $U$ made by the algorithm $\mathcal{A}_q$ and analyze the branches.
Using Stinespring dilation~\cite{watrous2018theory} and coherently retaining measurement outcomes and classical control, any adaptive $q$-query algorithm can be purified, with the known operations between successive oracle queries represented as isometries.
By preparing all ancillary registers at first, these isometries can be extended to unitaries.
Thus, without loss of generality, the algorithm's final pure state can be written as follows.
\begin{equation}\label{eq:algorithm-note-form}
    A^{(q+1)}U A^{(q)}\cdots A^{(2)}U A^{(1)}\ket{\mathrm{init}},
\end{equation}
Here, $\ket{\mathrm{init}}$ is the purified initial state, and $A^{(j)}$ is a known unitary specified by the known algorithm.

We record in a list $l$ which of the three terms is chosen in each query.
Since the phases depend only on the number of occurrences of each term, we group words according to the number $a$ of main blocks and the number $b$ of bad blocks.
For a sequence $l=(l_1,\ldots,l_q)\in\{\psi,G,B\}^q$, define
\begin{equation}
    Q_{l_j}=
    \begin{cases}
        \ketbra{\psi}{0}, & l_j=\psi,\\
        G_\Pi, & l_j=G,\\
        B_{\Pi,\psi}, & l_j=B,
    \end{cases}
\end{equation}
We also define
\begin{equation}
    \ket{\Psi_l}
    =A^{(q+1)}Q_{l_q}A^{(q)}\cdots A^{(2)}Q_{l_1}A^{(1)}\ket{\mathrm{init}}.
\end{equation}
Writing
\begin{equation}
    a(l)= |\{j:l_j=\psi\}|, \qquad b(l)= |\{j:l_j=B\}|,
\end{equation}
we further define
\begin{equation}\label{eq:Psiab-note}
    \ket{\Psi_{a,b}} =\sum_{\substack{l:\,a(l)=a\\b(l)=b}}\ket{\Psi_l}.
\end{equation}
Each sequence $l$ contributing to $\ket{\Psi_{a,b}}$ contains exactly $q - a - b$ occurrences of $G_\Pi$ and $b$ occurrences of $B_{\Pi, \psi}$ in Eq.~\eqref{eq:note-decomposition}. Hence, Eq.~\eqref{eq:algorithm-note-form} becomes
\begin{equation}
    \ket{\Psi(\theta,\phi)} =\sum_{a=0}^q\sum_{b=0}^{q-a} e^{i(q-a-b)\theta}e^{ib\phi}\ket{\Psi_{a,b}}.
\end{equation}
Averaging over $\theta$ and $\phi$ eliminates all cross terms between distinct pairs $a,b$.
\begin{align}\label{eq:sector-decomposition-main}
    \E_{\theta,\phi \in [0,2\pi)}\left[\proj{\Psi(\theta,\phi)}\right] &=\sum_{a=0}^q\sum_{b=0}^{q-a}\proj{\Psi_{a,b}} \\
    &=\rho_0+\rho_B,
\end{align}
Here, we have defined
\begin{equation}\label{eq:rho0-rhoB}
    \rho_0 \coloneqq \sum_{a = 0}^q \proj{\Psi_{a,0}}, \qquad \rho_B \coloneqq \sum_{a = 0}^q \sum_{b = 1}^{q - a} \proj{\Psi_{a,b}}.
\end{equation}
Below, we show that $\rho_0$, averaged over $\Pi$ and $G_\Pi$, can be reproduced from $q$ copies of $\ket{\psi}$.
We also upper-bound the average trace norm $\E_{\Pi,G,B}[\norm{\rho_B}_1]$ of $\rho_B$.

\begin{restatable}[Exact simulation of the $b=0$ sector]{lemma}{exactsimulationb}\label{lem:no-bad-simulation}
For any $q$-query algorithm, there exists a completely positive and trace-nonincreasing map $\mathfrak{S}_q^{(0)}$ independent of $\ket{\psi}$ satisfying
\begin{equation}\label{eq:no-bad-main-simulation}
    \mathfrak S_q^{(0)}(P_\psi^{\otimes q})
    =\E_{\Pi,G_\Pi}\!\left[\sum_{a=0}^q\proj{\Psi_{a,0}}\right].
\end{equation}
\end{restatable}
This map $\mathfrak S_q^{(0)}$ can be completed to a quantum channel by sending the missing trace to an arbitrary fixed failure state.
The proof idea of Lemma~\ref{lem:no-bad-simulation} is as follows.
First, choose a pair $(\Pi,G_\Pi)$ independent of $\ket{\psi}$ once and use the same pair in every query.
In each query, on the $\ket{0}$ component of the query, SWAP the next unused copy of $\ket{\psi}$ into the query register, and apply $G_\Pi$ to the $(I-\Pi)\ket{0}^{\perp}$ component.
Without measuring these two components along the way, send only the $\Pi\ket{0}^{\perp}$ component
to a failure branch.

On a successful branch in which the main block is used a total of $a$ times, the first $a$ copy registers become $\ket{0}$, while the remaining registers stay in $\ket{\psi}$.
Since this state does not depend on the positions at which the main block appears, the amplitudes of words with the same $a$ add coherently to give $\ket{\Psi_{a,0}}$.
On the other hand, since $\braket{0|\psi}=0$, the copy-register states corresponding to different values of $a$ are mutually orthogonal.
Thus, discarding the copy registers from the successful branches yields $\sum_{a=0}^q\proj{\Psi_{a,0}}$.
Finally, averaging over $(\Pi,G_\Pi)$ gives Eq.~\eqref{eq:no-bad-main-simulation}.
The details of the construction are given in Appendix~\ref{app:copy-simulator}.

\subsection{Bounding the weight of the remaining output component}~\label{sec:proof2}

The aim of this subsection is to upper-bound the average weight of the bad part $\rho_B$ in Eq.~\eqref{eq:sector-decomposition-main}.
Since $\rho_B\succeq0$, its trace norm equals its trace, and we need to bound the following quantity:
\begin{equation}\label{eq:rhoB-trace-weight}
    \E_{\Pi,G,B}\!\left[\norm{\rho_B}_1\right]
    =\E_{\Pi,G,B}\!\left[\Tr(\rho_B)\right]
    =\E_{\Pi,G,B}\!\left[
      \sum_{a=0}^q\sum_{b=1}^{q-a}\norm{\ket{\Psi_{a,b}}}^2
    \right]
\end{equation}

To interpret this quantity as a probability, we define the following auxiliary monitored experiment. Let $\mathcal A_q$ be an arbitrary quantum algorithm making $q$ queries with forward-only access.
Choose $\Pi,G_\Pi,B_{\Pi,\psi},\theta,\phi$ once according to Lemma~\ref{lem:haar-block-decomposition}, and use the same unitary $U_{\theta,\phi,\Pi,G,B}$ in Eq.~\eqref{eq:note-decomposition} for all $q$ oracle queries.

Replace each oracle query of $\mathcal A_q$ with a measurement operation having the following two Kraus operators, and record the outcome of each query:
\begin{equation}\label{eq:nonbad-bad-instrument}
    U^{(0)}_{\theta,\Pi,G}
    =Q_\psi+e^{i\theta}G_\Pi ({I - \Pi}) Q_0,
    \qquad
    U^{(1)}_{\phi,\Pi,B}
    =e^{i\phi}B_{\Pi,\psi}\Pi Q_0,
\end{equation}
where $Q_\psi=\ketbra{\psi}{0}$.
Let $\Pi_0$ denote the extension of $\Pi$ by zero on span $\{\ket{0}\}$.
All operators in the following identities act on $\cH$.
These operators satisfy
\begin{equation}
    (U^{(0)}_{\theta,\Pi,G})^\dagger
    U^{(0)}_{\theta,\Pi,G}
    = I_{\mathcal H} - \Pi_0,
    \qquad
    (U^{(1)}_{\phi,\Pi,B})^\dagger
    U^{(1)}_{\phi,\Pi,B}
    =\Pi_0,
    \qquad
    (U^{(0)}_{\theta,\Pi,G})^\dagger
    U^{(1)}_{\phi,\Pi,B}
    =0.
\end{equation}
We refer to outcome $0$ as no-$B$ and outcome $1$ as $B$.
We call the measurement process defined above the monitored experiment.
In this experiment, among the $q$ queries, define the probability of obtaining a $B$ outcome at least once as follows:
\begin{equation}\label{eq:p-hit-definition}
    p_{\mathrm{hit}}
    \coloneqq
    \Pr\!\left[\text{at least one of the $q$ queries yields $B$}\right]
\end{equation}

Here, the probability is averaged over the Haar randomness and the outcomes of the monitoring measurements.

Fix $\Pi,G,B,\theta$.
By the chain rule of probability and the Born rule for each measurement, the probability that all $q$ outcomes are no-$B$ is
\begin{align}
    &\Pr[0_1,\ldots,0_q\mid\theta,\Pi,G,B]\notag\\
    &\quad=
    \prod_{j=1}^q
    \Pr[0_j\mid0_1,\ldots,0_{j-1},\theta,\Pi,G,B]\notag\\
    &\quad=
    \norm{
      A^{(q+1)}U^{(0)}_{\theta,\Pi,G}A^{(q)}\cdots
      A^{(2)}U^{(0)}_{\theta,\Pi,G}A^{(1)}\ket{\mathrm{init}}
    }^2.
\end{align}
Expanding $U^{(0)}_{\theta,\Pi,G}$ as the sum of the main and good blocks, the expression above becomes
\begin{equation}
    \Pr[0_1,\ldots,0_q\mid\theta,\Pi,G,B]
    =\norm{
      \sum_{a=0}^q e^{i(q-a)\theta}\ket{\Psi_{a,0}}
    }^2
\end{equation}
Thus, averaging over the common phase $\theta$ gives
\begin{align}
    &\E_\theta\!\left[
      \Pr[0_1,\ldots,0_q\mid\theta,\Pi,G,B]
    \right]\notag\\
    &\quad=
    \sum_{a,a'=0}^q
    \E_\theta\!\left[e^{i(a'-a)\theta}\right]
    \braket{\Psi_{a',0}|\Psi_{a,0}}\notag\\
    &\quad=
    \sum_{a=0}^q\norm{\ket{\Psi_{a,0}}}^2.
\end{align}

On the other hand, the phase-averaged state in Eq.~\eqref{eq:sector-decomposition-main} has trace one, so
\begin{equation}
    1= \sum_{a=0}^q \norm{\ket{\Psi_{a,0}}}^2 + \sum_{a=0}^q\sum_{b=1}^{q-a} \norm{\ket{\Psi_{a,b}}}^2.
\end{equation}
Therefore, taking the complement of the event that all outcomes are no-$B$ and then averaging over $\Pi,G,B$, we obtain
\begin{equation}\label{eq:bad-weight-hit}
    p_{\mathrm{hit}} = 1 - \E_{\Pi,G,B}\!\left[
      \Pr[0_1,\ldots,0_q\mid\Pi,G,B]
    \right]
    =\E_{\Pi,G,B}\!\left[
      \sum_{a=0}^q\sum_{b=1}^{q-a}
      \norm{\ket{\Psi_{a,b}}}^2
    \right].
\end{equation}
That is, Eq.~\eqref{eq:rhoB-trace-weight} equals $p_{\mathrm{hit}}$.

For $m\in\{0,\ldots,q\}$, the probability that the first $m$ monitored oracle queries all yield no-$B$ is denoted by $s_m$.
For $m=0$, no queries have yet been made, and this condition is always satisfied, so $s_0=1$.
Furthermore, $s_q$ is the probability that all $q$ queries yield no-$B$.
That is, $s_q=1-p_{\mathrm{hit}}$.

Fix $m\in\{0,\ldots,q-1\}$.
Retain the branch in which the first $m$ outcomes are all no-$B$ without normalizing it, and run the algorithm until just before the ($m+1$)st oracle query.
This includes the known inter-query operation performed after the $m$th query. For $m=0$, however, only the initial state preparation before the first query is performed.
At this point, on the query register for the ($m+1$)st query, $\cH$, and all the remaining workspace registers, $\cM$, the output is a subnormalized state with average trace $s_m$.

The event that the first $m$ outcomes are no-$B$ splits into the event that the ($m+1$)st outcome is also no-$B$ and the event that the ($m+1$)st query yields $B$ for the first time.
Thus, writing $p_{m+1}^{\mathrm{first}}$ for the probability of the latter event, we have
\begin{equation}\label{eq: probability sequence}
    p_{m+1}^{\mathrm{first}}=s_m-s_{m+1}.
\end{equation}
Furthermore, $p_{m+1}^{\mathrm{first}}$ and $s_m$ satisfy the following relation.
\begin{restatable}{lemma}{pvssm}\label{lem:first-B-prob}
    For every $m \in \{0, 1, \ldots, q - 1\}$ satisfying $m(r - 1) < N - 1$, $p_{m+1}^{\mathrm{first}}$ and $s_m$ in Eq.~\eqref{eq: probability sequence} satisfy
    \begin{equation}
        p_{m+1}^{\mathrm{first}} \leq \frac{r-1}{(N - 1)-m(r - 1)}s_m.
    \end{equation}
\end{restatable}
The proof is deferred to Appendix~\ref{app:hidden-subspace}.

Using Eq.~\eqref{eq: probability sequence} and Lemma~\ref{lem:first-B-prob}, we obtain an upper bound on $p_{\mathrm{hit}}$.
\begin{restatable}{lemma}{firstBbound}\label{lem:first-B}
    If $q(r - 1)<N-1$, the probability in Eq.~\eqref{eq:p-hit-definition} satisfies
    \begin{equation}\label{eq:first-B-bound}
        p_{\mathrm{hit}}
        \leq\frac{q(r - 1)}{N-1}.
    \end{equation}
\end{restatable}
\begin{proof}
    By Eq.~\eqref{eq: probability sequence}, $p_{m+1}^{\mathrm{first}}=s_m-s_{m+1}$, so
    \begin{align}
        s_{m+1} &= s_m - p_{m+1}^{\mathrm{first}} \\
        &\geq \left(1- \frac{r-1}{(N - 1)-m(r - 1)}\right)s_m \\
        &= \frac{(N - 1) -(m+1)(r-1)}{(N - 1)-m(r - 1)}s_m.
    \end{align}
    Iterating from $s_0=1$ gives the following product
    \begin{equation}
        s_q
        \geq
        \prod_{m=0}^{q-1}
        \frac{(N - 1) -(m+1)(r-1)}{(N - 1)-m(r - 1)}
        =\frac{(N - 1) - q(r-1)}{N - 1}.
    \end{equation}
    Therefore,
    \begin{equation}
        p_{\mathrm{hit}}=1-s_q\leq\frac{q(r-1)}{N-1}.
    \end{equation}
\end{proof}

This lemma allows us to upper-bound the weight of the sectors containing the bad block.
\begin{restatable}[Weight of the sectors containing $B$]{lemma}{badblockbound}\label{lem:bad-sector-bound}
If $q(r-1)<N-1$, then
\begin{equation}\label{eq:bad-sector-main-bound}
    \E_{\Pi,G,B}\!\left[
      \sum_{a=0}^q\sum_{b=1}^{q-a}\norm{\ket{\Psi_{a,b}}}^2
    \right]
    \leq\frac{q(r-1)}{N-1}.
\end{equation}
\end{restatable}

\begin{proof}
By Eq.~\eqref{eq:bad-weight-hit}, the total weight of the sectors with $b\neq0$ equals the probability that the monitored process enters the hidden subspace $\Pi \ket{0}^{\perp}$ at least once.
Lemma~\ref{lem:first-B} gives
\begin{equation}
    \E_{\Pi,G,B}\!\left[
      \sum_{a=0}^q\sum_{b=1}^{q-a}
      \norm{\ket{\Psi_{a,b}}}^2
    \right] = p_{\mathrm{hit}} \leq\frac{q(r-1)}{N-1}.
\end{equation}
\end{proof}

We now have all the ingredients needed to prove Theorem~\ref{thm:copy-reduction}.
We complete its proof below.

\begin{proof}[Proof of Theorem~\ref{thm:copy-reduction}]
For readability, we write
\begin{equation}
    \bar\rho_0=\E_{\Pi,G}[\rho_0],
    \qquad
    \bar\rho_B=\E_{\Pi,G,B}[\rho_B].
\end{equation}
By Eq.~\eqref{eq:sector-decomposition-main}, the Haar-averaged output of the original quantum algorithm is $\bar\rho_0+\bar\rho_B$.
By Lemma~\ref{lem:no-bad-simulation}, the simulator using $q$ copies can output $\bar\rho_0$ on its success branch.
Thus, the trace deficit on this input is $1-\Tr\bar\rho_0=\Tr\bar\rho_B$.
Completing the map to a channel by outputting a fixed state $\tau$ on failure gives the output $\bar\rho_0+\Tr(\bar\rho_B)\tau$.
Therefore,
\begin{align}
    \Dtr(\bar\rho_0+\bar\rho_B,\bar\rho_0+\Tr(\bar\rho_B)\tau)
    &=\frac12\norm{\bar\rho_B-\Tr(\bar\rho_B)\tau}_1\\
    &\leq \Tr(\bar\rho_B)\\
    & = \E_{\Pi,G,B}\!\left[
      \sum_{a=0}^q\sum_{b=1}^{q-a}
      \norm{\ket{\Psi_{a,b}}}^2
    \right] \\
    &\leq\frac{q(r-1)}{N-1},
\end{align}
where the last inequality follows from Lemma~\ref{lem:bad-sector-bound}.
This proves the theorem.
\end{proof}

\section{Lower bounds for quantum PAC learning with forward unitary queries}\label{sec:pac-consequences}
In this section, we establish query lower bounds for realizable and agnostic PAC learning in the model with forward-only access to a unitary that generates the quantum examples defined in Section~\ref{sec:pac-learning}.
Both proofs use Theorem~\ref{thm:copy-reduction} to approximately reproduce the Haar-averaged output of a learner with forward access from $q$ copies of a quantum example.
This allows us to apply the known lower bounds~\cite{arunachalam2018optimal} in the copy model of quantum examples.
We consider realizable learning in Section~\ref{sec:realizable}
and agnostic learning in Section~\ref{sec:agnostic}.

In the hard instances constructed below, we choose an oracle space and a known representation in which the initial state $|0\rangle$ is orthogonal to all example basis states.
Such a representation exists in every sufficiently large permitted dimension and is fixed independently of the unknown concept and distribution.
The known isometric encoding preserves the copy complexity.
Consequently, the signal subspaces used below satisfy $\mathcal S\subseteq|0\rangle^\perp$, as required by Theorem~\ref{thm:copy-reduction}.

\subsection{Forward-query lower bounds for realizable learning}\label{sec:realizable}
Let $d=\vc(\cC)\geq2$, and fix a set shattered by $\cC$ consisting of $d$ points $x_1,\ldots,x_d$. When a distribution $D$ places mass only on these points, the quantum example for the realizable learning problem $(c,\mathcal D)$ has the following form:
\begin{equation}
    \ket{\psi_{\mathcal{D},c}}
    =\sum_{i=1}^{d}\sqrt{\mathcal{D}(x_i)}\,\ket{x_i,c(x_i)}.
\end{equation}

For a given learning problem $(c,\mathcal{D})$, the label of each $x_i$ is fixed, but as the learning problem varies, $c(x_i)$ can be either $0$ or $1$.
Therefore, we can choose the following signal subspace containing all the quantum examples for these learning problems:
\begin{equation}
    \cS\coloneqq \operatorname{span}
    \bigl\{\ket{x_i,y}:i\in[d],\ y\in\{0,1\}\bigr\},
    \qquad r\coloneqq\dim\cS\leq2d<\infty.
\end{equation}
Importantly, the dimension $r$ of the subspace is bounded without depending on the ambient dimension $N$.
This allows us to make the simulation error in Theorem~\ref{thm:copy-reduction} small by choosing sufficiently large $N$.

Apply part (i) of Lemma~\ref{lem:qsample-copy-bound} to the set fixed above.
Consider learners that, on every realizable $(c,\mathcal D)$, have failure probability at most $\delta$ ($0<\delta\leq3/8$).
Any such learner using $q$ copies must satisfy the following bound for some $\delta$-independent constant $c_0>0$:
\begin{equation}\label{eq:realizable-known-copy-bound}
    q\geq c_0\,\frac{d+\log(1/\delta)}{\varepsilon}.
\end{equation}
Below, we use Theorem~\ref{thm:copy-reduction} to transfer this copy lower bound to forward access.

\begin{theorem}[Realizable PAC complexity with forward-only access]
\label{thm:realizable-main}
Let $\cC$ be a Boolean concept class with finite VC dimension $d=\vc(\cC)\geq2$, and let $0<\varepsilon\leq1/32$ and $0<\delta\leq1/4$.
In the forward-only access model of Definition~\ref{def:state-preparation-access}, any realizable $(\varepsilon,\delta)$-PAC learner requires
\begin{equation}\label{eq:realizable-lb}
    q=\Omega\!\left(
        \frac{d+\log(1/\delta)}{\varepsilon}
    \right)
\end{equation}
queries in the worst case over permitted ambient dimensions, realizable learning instances, and compatible unitary completions.
\end{theorem}

\begin{proof}
Fix a realizable learner family $(\mathcal A_N)_N$ with finite worst-case query count $q$, as discussed in Section~\ref{sec:samplecomplexity summerize}. 
Write the dimension of the fixed signal subspace $\cS$ as $r$.
Choose a permitted ambient dimension $N$ sufficiently large that
\begin{equation}
    \frac{q(r-1)}{N-1}\leq\frac{\delta}{2}.
\end{equation}
Consider the corresponding learner $\mathcal A_N$.
It uses at most $q$ queries and succeeds with probability at least $1-\delta$ for every realizable instance and every compatible unitary completion in this dimension.
By adding dummy queries when necessary, we may regard it as making exactly $q$ queries.

Fix any learning problem $(c,\mathcal D)$ in the restricted family above, and let its quantum example be $\ket{\psi}=\ket{\psi_{\mathcal{D},c}}\in\cS$.
The assumed learner succeeds whenever $U\ket{0}=\ket{\psi}$, for every state-preparation unitary $U$.
Therefore, even after choosing $U$ as a Haar-random completion and averaging, its success probability is at least $1-\delta$.

Theorem~\ref{thm:copy-reduction} provides a simulator that, given $q$ copies of $\ket{\psi}$, reproduces this averaged output.
Since the trace distance between the two outputs is at most $q(r-1)/(N-1)$, the difference between the probabilities of the success event that the output hypothesis has error at most $\varepsilon$ is also at most $q(r-1)/(N-1)$.
Hence, the success probability of the copy simulator is bounded below as follows:
\begin{equation}
    1-\delta-\frac{q(r-1)}{N-1}
    \geq 1-\frac{3\delta}{2}.
\end{equation}

Since the simulator in the theorem does not depend on $\ket{\psi}$, this success guarantee holds simultaneously for all $(c,\mathcal D)$ in the restricted family above.
Setting $\delta'=3\delta/2\leq3/8$, we can apply Eq.~\eqref{eq:realizable-known-copy-bound}.
Moreover, for $\delta\leq1/4$, we have $\log(1/\delta')\geq\frac12\log(1/\delta)$, which gives the following bound:
\begin{equation}
    q\geq c_0\,\frac{d+\log(1/\delta')}{\varepsilon}
    \geq\frac{c_0}{2}\,\frac{d+\log(1/\delta)}{\varepsilon}.
\end{equation}
This completes the proof of Theorem~\ref{thm:realizable-main}.
\end{proof}

\subsection{Forward-query bounds for agnostic learning}\label{sec:agnostic}
Let $d=\vc(\cC)\geq1$, and fix a set shattered by $\cC$ consisting of the points $x_1,\ldots,x_d$.
The quantum example for any distribution $\cD$ whose inputs are restricted to these points is
\begin{equation}
    \ket{\psi_{\cD}}
    =\sum_{i=1}^{d}\sum_{y\in\{0,1\}}
      \sqrt{\cD(x_i,y)}\,\ket{x_i,y}.
\end{equation}
Thus, as in the realizable case in Section~\ref{sec:realizable}, even as the distribution $\cD$ varies, the quantum examples belong to the following fixed signal subspace:
\begin{equation}
    \cS\coloneqq\operatorname{span}
    \bigl\{\ket{x_i,y}:i\in[d],\ y\in\{0,1\}\bigr\},
    \qquad r\coloneqq\dim\cS\leq2d<\infty.
\end{equation}

Part (ii) of Lemma~\ref{lem:qsample-copy-bound} applies to distributions supported on the set fixed above, even when improper outputs are allowed.
Consider learners that, for every $\cD$, have failure probability at most $\delta$ ($0<\delta\leq3/8$).
Any such learner using $q$ copies must satisfy the following bound for some $\delta$-independent constant $c_0>0$:
\begin{equation}\label{eq:agnostic-known-copy-bound}
    q\geq c_0\,\frac{d+\log(1/\delta)}{\varepsilon^2}.
\end{equation}

\begin{theorem}[Agnostic PAC complexity with forward-only access]
\label{thm:agnostic-main}
Let $\cC$ be a Boolean concept class with finite VC dimension $d=\vc(\cC)\geq1$, and let $0<\varepsilon\leq1/16$ and $0<\delta\leq1/4$.
In the forward-only access model of Definition~\ref{def:state-preparation-access}, any $(\varepsilon,\delta)$-agnostic PAC learner requires
\begin{equation}\label{eq:agnostic-lb}
    q=\Omega\!\left(
        \frac{d+\log(1/\delta)}{\varepsilon^2}
    \right)
\end{equation}
queries in the worst case over permitted ambient dimensions, joint example distributions, and compatible unitary completions.
\end{theorem}

\begin{proof}
Fix an agnostic learner family $(\mathcal A_N)_N$ with finite worst-case query count $q$, as defined in Section~\ref{sec:samplecomplexity summerize}.
For the fixed signal subspace of dimension $r$, choose a permitted ambient dimension $N$ sufficiently large that $q(r-1)/(N-1)\leq\delta/2$.
Consider the corresponding learner $\mathcal A_N$.
It uses at most $q$ queries and succeeds with probability at least $1-\delta$ for every joint example distribution and every compatible unitary completion in this dimension.
By adding dummy queries when necessary, we may regard it as making exactly $q$ queries.

Fix any distribution $\cD$ in the restricted family above.
Since this learner succeeds for every unitary satisfying $U\ket{0}=\ket{\psi_{\cD}}$, its success probability is at least $1-\delta$ even after averaging with respect to Haar measure.
By Theorem~\ref{thm:copy-reduction}, the trace distance from the output of the simulator using $q$ copies is at most $q(r-1)/(N-1)\leq\delta/2$, so the failure probability of the simulator is at most $\delta'=3\delta/2$.
This simulator does not depend on $\ket{\psi_{\cD}}$, and the same guarantee holds for every $\cD$.
Since $\delta'\leq3/8$, we can apply Eq.~\eqref{eq:agnostic-known-copy-bound}.
Furthermore, $\log(1/\delta')\geq\frac12\log(1/\delta)$ implies Eq.~\eqref{eq:agnostic-lb}.
\end{proof}

\acknowledgments{
\begin{sloppypar}
This work was supported by JST BOOST, Japan Grant Number JPMJBS2418, MEXT Quantum Leap Flagship Program (MEXT QLEAP) JPMXS0118069605 and JPMXS0120351339; Japan Science and Technology Agency (JST) as part of Adopting Sustainable Partnerships for Innovative Research Ecosystem (ASPIRE), Grant Number JPMJAP25A3; JST CREST, Grant Number JPMJCR25I5; JST NEXUS, Grant Number JPMJNX26C9; JSPS KAKENHI Grant No. 23K21643 and 26K25550; and IBM Quantum.
\end{sloppypar}
}

\section*{Author contributions}
N. I. led the research and prepared the manuscript draft.
S. Y. and M. M. contributed to discussions and preparation of the manuscript.
GPT-5.6 sol Pro was used to assist with manuscript preparation and the proof of Lemma~\ref{lem:first-B-prob}.
The authors independently verified the mathematical arguments and take responsibility for the content of the manuscript.

\bibliographystyle{unsrtnat}
\bibliography{main}

\appendix

\section{Choi representation}\label{subsec:choi}
For a linear map $\Phi: \mathsf{L}(\mathcal{I}) \to \mathsf{L}(\mathcal{O})$, the Hilbert--Schmidt adjoint $\Phi^\dagger$ is uniquely determined as the following linear map:
\begin{equation}
\Phi^\dagger: \mathsf{L}(\mathcal{O}) \longrightarrow \mathsf{L}(\mathcal{I})
\end{equation}
This map satisfies
\begin{equation}\label{eq:superoperator-adjoint}
\Tr\!\left[Y^\dagger\Phi(X)\right] = \Tr\! \left[ \bigl(\Phi^\dagger(Y)\bigr)^\dagger X \right]
\end{equation}
for all $X\in\mathsf{L}(\mathcal{I})$ and $Y\in\mathsf{L}(\mathcal{O})$.

\paragraph{Choi operator.}
Consider the linear map
\begin{equation}
\Phi: \mathsf{L}(\mathcal{I}) \longrightarrow \mathsf{L}(\mathcal{O})
\end{equation}
Its Choi operator (also called its Choi matrix) is defined by
\begin{equation}\label{eq:choi-definition}
J_\Phi = \sum_{i,j} \ketbra{i}{j}_{\mathcal{I}} \otimes \Phi\!\left(\ketbra{i}{j}\right)_{\mathcal{O}}.
\end{equation}
Thus,
\begin{equation}
J_\Phi
\in
\mathsf{L}(\mathcal{I}\otimes\mathcal{O}).
\end{equation}

The map $\Phi$ can be recovered from its Choi operator by
\begin{equation}\label{eq:choi-reconstruction}
\Phi(X) = \Tr_{\mathcal{I}}\!\left[ \left( X^\top\otimes I_{\mathcal{O}} \right)J_\Phi \right].
\end{equation}
The transpose in Eq.~\eqref{eq:choi-reconstruction} is taken with respect to the fixed computational basis used in Eq.~\eqref{eq:choi-definition}.

An equivalent identity is the following Choi pairing formula:
\begin{equation}\label{eq:choi-pairing}
\Tr\!\left[ Y\Phi(X) \right] = \Tr\!\left[ J_\Phi \left( X^\top\otimes Y \right) \right].
\end{equation}
This identity holds for all $X\in\mathsf{L}(\mathcal{I})$ and $Y\in\mathsf{L}(\mathcal{O})$.

The Choi representation is useful for characterizing complete positivity and trace conditions for a linear map $\Phi:\mathsf{L}(\mathcal{I})\to\mathsf{L}(\mathcal{O})$.
In particular, $\Phi$ is completely positive if and only if its Choi operator is positive semidefinite:
\begin{equation}\label{eq:choi-cp}
\Phi\text{ is completely positive}
\quad\Longleftrightarrow\quad
J_\Phi\succeq0.
\end{equation}
A completely positive map is trace-preserving if and only if
\begin{equation}\label{eq:choi-cptp}
\Tr_{\mathcal{O}}\!\left[J_\Phi\right] = I_{\mathcal{I}},
\end{equation}
and it is trace-nonincreasing if and only if
\begin{equation}\label{eq:choi-cptni}
\Tr_{\mathcal{O}}\!\left[J_\Phi \right] \preceq I_{\mathcal{I}}.
\end{equation}

\paragraph{Choi vector of a linear operator.}

For a linear operator
\begin{equation}
V: \mathcal{I} \longrightarrow \mathcal{O},
\end{equation}
we define its Choi vector by
\begin{equation}\label{eq:choi-vector}
\ket{V}\!\rangle_{\mathcal{I}\mathcal{O}} \coloneqq \sum_i \ket{i}_{\mathcal{I}} \otimes V\ket{i}_{\mathcal{I}}.
\end{equation}
Define the linear map $\Phi_V$ by
\begin{equation}
    \Phi_V(X) = VXV^\dagger.
\end{equation}
The Choi operator of $\Phi_V$ is defined as the following rank-one operator:
\begin{equation}\label{eq:rank-one-choi}
    J_{\Phi_V} \coloneqq \ket{V}\!\rangle\!\langle\!\bra{V}.
\end{equation}

\section{Proof of Lemma~\ref{lem:haar-block-decomposition}}\label{app:haar-decomposition}
In this section, we prove Lemma~\ref{lem:haar-block-decomposition}.
\haardecomposition*

\begin{proof}[Proof of Lemma~\ref{lem:haar-block-decomposition}]
Let $W:\ket{0}^{\perp}\to\ket{\psi}^{\perp}$ in the following decomposition be a Haar-random isometry:
\begin{equation}
    U=\ketbra{\psi}{0}+WQ_0.
\end{equation}
Let $\cB_\psi = \cS \cap \ket{\psi}^\perp$, and denote the projector onto $W^\dagger\cB_\psi$ by $\Pi$.

Since $W$ is Haar-random, $\Pi$ is uniformly distributed over all projectors on $\ket{0}^{\perp}$ of rank $(r-1)$.
Conditioning on a fixed $\Pi$, orthogonality gives
\begin{equation}
    W(I-\Pi)\ket{0}^{\perp}=\cS^{\perp}, \qquad W\Pi\ket{0}^{\perp}=\cB_\psi,
\end{equation}
because $\ket{\psi}^\perp = \cB_\psi \oplus \cS^\perp$.
Thus, the action of the stabilizer subgroup on the two orthogonal pairs of domain and range blocks, $(I - \Pi)\ket{0}^\perp, \cS^\perp$ and $\Pi\ket{0}^\perp, \cB_\psi$, implies that the following two restrictions are independent Haar-random isometries:
\begin{equation}
    G_\Pi \coloneqq W|_{(I - \Pi)\ket{0}^\perp}, \qquad B_{\Pi,\psi} \coloneqq W|_{\Pi\ket{0}^\perp}.
\end{equation}
The distributions of $\Pi$ and $G_\Pi$ depend only on the fixed space $\cS^{\perp}$ and are therefore independent of $\psi$.
Finally, the Haar distributions of $G_\Pi$ and $B_{\Pi,\psi}$ are invariant under multiplication by independent phases, yielding Eq.~\eqref{eq:note-decomposition}.
\end{proof}

\section{Proof of Lemma~\ref{lem:no-bad-simulation}}\label{app:copy-simulator}
\exactsimulationb*
\begin{proof}[Proof of Lemma~\ref{lem:no-bad-simulation}]
We first construct a simulator for a fixed pair $(\Pi,G_\Pi)$ in Eq.~\eqref{eq:note-decomposition}, where $G_\Pi:(I-\Pi)\ket{0}^{\perp}\to\cS^{\perp}$.
Introduce the query register $\mathsf{Q}$ of the original algorithm, copy registers $\mathsf{C}_1,\ldots,\mathsf{C}_q$, a counter register with basis $\ket{0},\ldots,\ket{q}$ denoted by $\mathsf R$, and a failure qubit $\mathsf{F}$. Initialize the copy registers in $\ket{\psi}^{\otimes q}$, the counter in $\ket{0}$, and the failure qubit in $\ket{0}$.
The query register $\mathsf{Q}$ is the register to which each query to $U$ is applied in the original algorithm.

For $0\leq c\leq q$, define
\begin{equation}\label{eq:Pc-note}
    \pi_c \coloneqq \bigotimes_{j=1}^{c}P_0^{\mathsf C_j} \otimes \bigotimes_{j=c+1}^{q}P_{\cS}^{\mathsf C_j},
\end{equation}
where $\pi_{0} = \bigotimes_{j=1}^{q}P_{\cS}^{\mathsf C_j}$ and $\pi_{q} = \bigotimes_{j=1}^{q}P_0^{\mathsf C_j}$.
For $0 \leq c < q$, set
\begin{align}
    V_c^\psi &=\operatorname{SWAP}_{\mathsf Q,\mathsf C_{c+1}} (P_0^{\mathsf Q}\otimes \pi_c) \otimes\ketbra{c+1}{c}_{\mathsf R} \otimes P_0^{\mathsf F},\label{eq:Vc-psi-note}\\
    V_c^G &=G_\Pi(I-\Pi)Q_0^{\mathsf Q} \otimes \pi_c \otimes\ketbra{c}{c}_{\mathsf R} \otimes P_0^{\mathsf F},\label{eq:Vc-G-note}\\
    V_c^B &=\Pi Q_0^{\mathsf Q} \otimes \pi_c \otimes\ketbra{c}{c}_{\mathsf R} \otimes\ketbra{1}{0}_{\mathsf F}.\label{eq:Vc-B-note}
\end{align}
The operator $V_c^B$ merely records the bad block and does not attempt to reproduce the unknown map $B_{\Pi,\psi}$.

Set
\begin{equation}
    K_0=\sum_{c=0}^{q-1}(V_c^\psi+V_c^G), \qquad K_1=\sum_{c=0}^{q-1}V_c^B,
\end{equation}
and further define
\begin{equation}
    P_{\mathrm{domain}} =\sum_{c=0}^{q-1} I_{\mathsf Q}\otimes \pi_c\otimes\ketbra{c}{c}_{\mathsf R} \otimes P_0^{\mathsf F}.
\end{equation}
The three input subspaces in the query register $\mathsf{Q}$, namely,
the spaces corresponding to $\ket{0}$, $(I-\Pi)\ket{0}^\perp$, and $\Pi \ket{0}^\perp$, are mutually orthogonal by the relation
\begin{equation}
    P_0+(I-\Pi)Q_0+\Pi Q_0=I.
\end{equation}
The outputs are also mutually orthogonal. The map $V_c^\psi$ increments the counter and has its query-register output in $\cS$. The map $V_c^G$ leaves the counter unchanged and has its query-register output in $\cS^\perp$. The map $V_c^B$ sets register $\mathsf F$ to $\ket{1}$.
Outputs corresponding to different values of $c$ are separated by the counter, the pattern of states in the copy registers, or the orthogonal decomposition $\cS\oplus\cS^\perp$.
Therefore,
\begin{align}\label{eq:copy-completeness-note}
    K_0^\dagger K_0 &= \sum_{c=0}^{q- 1} \left( (V_c^\psi)^\dagger V_c^\psi + (V_c^G)^\dagger V_c^G \right)\\ \notag
    & = \sum_{c=0}^{q- 1} (P_0^{\mathsf Q}+(I-\Pi)Q_0^{\mathsf Q})\otimes \pi_c\otimes\ketbra{c}{c}_{\mathsf R} \otimes P_0^{\mathsf F} \preceq I,
\end{align}
\begin{equation}
    K_1^\dagger K_1 = \sum_{c=0}^{q- 1} \Pi Q_0^{\mathsf Q} \otimes \pi_c \otimes\ketbra{c}{c}_{\mathsf R} \otimes P_0^{\mathsf{F}},
\end{equation}
\begin{equation}
    K_0^\dagger K_0 + K_1^\dagger K_1 = P_{\mathrm{domain}}.
\end{equation}
By adding Kraus operators supported on $I-P_{\mathrm{domain}}$ and mapping into the failure subspace spanned by $\ket{1}_\mathsf{F}$, we can complete $K_0$ and $K_1$ to a trace-preserving channel.
This completion depends on the fixed known pair $(\Pi,G_\Pi)$ but not on $\ket{\psi}$.

Extend $A^{(j)}$ by the identity on $\mathsf C_1\cdots\mathsf C_q\mathsf R\mathsf F$, denote this extension by $\widetilde A^{(j)}$, and define the successful $q$-round operator by
\begin{equation}\label{eq:fixed-pair-success-operator}
    W_{\Pi,G}^{(0)} =\widetilde A^{(q+1)}K_0\widetilde A^{(q)} \cdots K_0\widetilde A^{(2)}K_0\widetilde A^{(1)}.
\end{equation}
For an arbitrary state $\rho$ on the copy registers, define
\begin{equation}\label{eq:fixed-pair-success-map}
    \mathcal C_{\Pi,G}^{(0)}(\rho) =\Tr_{\mathsf{C}_1\cdots\mathsf{C}_q\mathsf{R} \mathsf{F}} \left[ W_{\Pi,G}^{(0)} \left( P_{\mathrm{init}}\otimes \rho \otimes P_0^{\mathsf{R}}\otimes P_0^{\mathsf{F}} \right) (W_{\Pi,G}^{(0)})^\dagger
      \right].
\end{equation}
By Eq.~\eqref{eq:copy-completeness-note}, this map is completely positive and trace-nonincreasing.
It is the success branch of the full simulator, which reuses the same pair $(\Pi,G_\Pi)$ in every simulated query.

Define
\begin{equation}\label{eq:Gamma-note}
    \ket{\Gamma_c^\psi} \coloneqq \ket{0}^{\otimes c}\otimes\ket{\psi}^{\otimes(q-c)}.
\end{equation}
For a successful history with counter value $c$, the state of the copy registers is exactly $\ket{\Gamma_c^\psi}$.
The map $V_c^\psi$ implements $\ketbra{\psi}{0}$ by swapping the next unused copy into the query register, while $V_c^G$ implements the fixed good block $G_\Pi$.
After $q$ queries, we obtain
\begin{align}\label{eq:copy-success-note}
    &W_{\Pi,G}^{(0)} \left( \ket{\mathrm{init}} \otimes\ket{\psi}^{\otimes q}_{\mathsf C_1\cdots\mathsf C_q} \otimes\ket{0}_{\mathsf R}\otimes\ket{0}_{\mathsf F} \right)\nonumber\\
    &\qquad= \sum_{a=0}^q \ket{\Psi_{a,0}(\Pi,G)}\otimes \ket{\Gamma_a^\psi}_{\mathsf C_1\cdots\mathsf C_q} \otimes\ket{a}_{\mathsf R}\otimes\ket{0}_{\mathsf F}.
\end{align}
Here, to make explicit the dependence of $\ket{\Psi_{a,0}}$ on the fixed pair $(\Pi,G_\Pi)$, we have written $\ket{\Psi_{a,0}(\Pi,G)}$.
All sequences in which the number of $\psi$ blocks equals the same value $a$ share the auxiliary state $\ket{\Gamma_a^\psi}_{\mathsf C_1\cdots\mathsf C_q} \otimes\ket{a}_{\mathsf R}\otimes\ket{0}_{\mathsf F}$ and are therefore added coherently.
Different values of $a$ are orthogonal in the counter register.
Hence,
\begin{equation}\label{eq:fixed-pair-sector-simulation}
    \mathcal{C}_{\Pi,G}^{(0)}(P_\psi^{\otimes q}) =\sum_{a=0}^q\proj{\Psi_{a,0}(\Pi,G)}.
\end{equation}

Denote the joint probability measure of $(\Pi,G_\Pi)$ by $\mu$.
Averaging the full $q$-round success maps, we define a single map by
\begin{equation}\label{eq:averaged-success-map}
    \mathfrak{S}_q^{(0)}(\rho)
    =\int \mathcal{C}_{\Pi,G}^{(0)}(\rho)\,d\mu(\Pi,G).
\end{equation}
The integral in Eq.~\eqref{eq:averaged-success-map} averages the map $\mathcal C_{\Pi,G}^{(0)}$ over the pair shared by all $q$ queries, namely, $(\Pi,G_\Pi)$. Thus, it is not the composition of a single-query Haar-averaged channel $q$ times.

To verify that Eq.~\eqref{eq:averaged-success-map} defines a valid operation, denote the Choi operator of $\mathcal C_{\Pi,G}^{(0)}$ by $J_{\Pi,G}$.
Since each map is completely positive and trace-nonincreasing, we have
\begin{equation}
    J_{\Pi,G}\succeq0, \qquad \Tr_{\mathrm{out}}J_{\Pi,G}\preceq I_{\mathrm{in}}.
\end{equation}
Therefore,
\begin{equation}
    \overline J \coloneqq \int J_{\Pi,G}\,d\mu(\Pi,G)
    \succeq0, \qquad \Tr_{\mathrm{out}}\overline J\preceq I_{\mathrm{in}}.
\end{equation}
Thus, $\overline J$ is the Choi operator of a completely positive, trace-nonincreasing map $\mathfrak{S}_q^{(0)}$.
Since the measure $\mu$ and each map $\mathcal C_{\Pi,G}^{(0)}$ are independent of $\ket{\psi}$, dependence on the unknown state arises only through the input $P_\psi^{\otimes q}$.
Averaging Eq.~\eqref{eq:fixed-pair-sector-simulation} gives
\begin{equation}
    \mathfrak{S}_q^{(0)}(P_\psi^{\otimes q}) =\E_{\Pi,G}\!\left[ \sum_{a=0}^q\proj{\Psi_{a,0}} \right],
\end{equation}
which is Eq.~\eqref{eq:no-bad-main-simulation}.
\end{proof}

\section{Proof of Lemma \ref{lem:first-B-prob}}\label{app:hidden-subspace}
In this appendix, we prove Lemma~\ref{lem:first-B-prob}.
\pvssm*
To simplify the notation, set
\begin{equation}
    n \coloneqq N-1, \quad k \coloneqq r-1, \quad d_0 \coloneqq n-k=\dim{\cS^{\perp}}.
\end{equation}
On $ \ket{0}^{\perp} $, for a rank-$d_0$ projector $I - \Pi$, choose $G_{\Pi}:(I - \Pi) \ket{0}^{\perp} \to{\cS^{\perp}}$ Haar-randomly and extend it by zero on $\Pi  \ket{0}^{\perp} $.
Then,
\begin{equation}\label{eq:input-output space of G}
    G_{\Pi}^\dagger G_{\Pi}= I - \Pi, \qquad G_{\Pi}G_{\Pi}^\dagger=I_{{\cS^{\perp}}}.
\end{equation}
Using the input-first convention for Choi vectors in Eq.~\eqref{eq:choi-vector}, write
\begin{equation}
    J(G_{\Pi})
    :=J_{\Phi_{G_{\Pi}}}
    =\ket{G_{\Pi}}\!\rangle\!\langle\!\bra{G_{\Pi}}.
\end{equation}
For each integer $g\geq0$, define
\begin{equation}\label{eq:Gamma-g-def}
    \Gamma_g(I - \Pi)=\E_{G_{\Pi}}[J(G_{\Pi})^{\otimes g}].
\end{equation}
When $g=0$, the tensor power in Eq.~\eqref{eq:Gamma-g-def} is the scalar $1$, that is, $J(G_{\Pi})^{\otimes g} = 1$.

The following lemma decomposes the final state of an adaptive circuit using $G_\Pi$ into a sum indexed by the number of uses of $G_\Pi$.
Note that $G_\Pi (I - \Pi)Q_0 = G_\Pi$.
\begin{lemma}\label{lem:homogeneous-network}
Consider a purified quantum algorithm that queries $U^{(0)}_{\theta,\Pi,G}$ in Eq.~\eqref{eq:nonbad-bad-instrument} a total of $m$ times. There exist, independently of $G$ and $\Pi$, linear maps $L_{m,g}$ such that the final state can be written as
\begin{equation}\label{eq:homogeneous-expansion}
    \ket{\Phi_m(\theta,G_\Pi, \Pi)}
    =\sum_{g=0}^m e^{ig\theta}
      L_{m,g}\ket{G_\Pi}\!\rangle^{\otimes g}.
\end{equation}
The map $L_{m,g}$ may depend on the known network and $\psi$, but not on $G_\Pi$ or $\Pi$.
\end{lemma}

\begin{proof}
For $\mathsf L( \ket{0}^{\perp} ,{\cS^{\perp}})$, choose a matrix-unit basis $\{E_\alpha\}_\alpha$ and write
\begin{equation}
    G_{\Pi}=\sum_\alpha g_\alpha E_\alpha.
\end{equation}
For a quantum algorithm that queries $U_{\theta,\Pi,G}^{(0)}$ a total of $m$ times, the final pure state can be written as
\begin{align}
    \ket{\Phi_m(\theta,G_\Pi, \Pi)} &= A^{(m + 1)}U_{\theta,\Pi,G}^{(0)}A^{(m)}\cdots U_{\theta,\Pi,G}^{(0)}A^{(2)}U_{\theta,\Pi,G}^{(0)}A^{(1)}\ket{\mathrm{init}} \\
    &= A^{(m + 1)}\left(Q_\psi+e^{i\theta}G_\Pi ({I - \Pi} )Q_0 \right)A^{(m)}\cdots \\
    &\quad \times \left(Q_\psi+e^{i\theta}G_\Pi ({I - \Pi} )Q_0 \right)A^{(2)}\left(Q_\psi+e^{i\theta}G_\Pi ({I - \Pi} )Q_0 \right)A^{(1)}\ket{\mathrm{init}}
\end{align}
Expanding this expression, the sum of all words in which $G_{\Pi}$ occurs exactly $g$ times is a homogeneous polynomial in the components $g_\alpha$ of degree $g$, and therefore has the form
\begin{equation}
    \sum_{\alpha_1,\ldots,\alpha_g}
      g_{\alpha_1}\cdots g_{\alpha_g}
      \ket{v_{\alpha_1,\ldots,\alpha_g}}.
\end{equation}
Here, $\ket{v_{\alpha_1,\ldots,\alpha_g}}$ is obtained by considering all subsets of the $m$ queries of size $g$, replacing $G_\Pi$ at the selected positions by $E_{\alpha_1},\ldots,E_{\alpha_g}$ in chronological order, replacing the remaining $m-g$ queries by $Q_\psi$, and summing the resulting vectors.

Define $L_{m,g}$ on the tensor product basis of Choi vectors by
\begin{equation}
    L_{m,g}
    \bigl(
      \ket{E_{\alpha_1}}\!\rangle\otimes\cdots\otimes
      \ket{E_{\alpha_g}}\!\rangle
    \bigr)
    =\ket{v_{\alpha_1,\ldots,\alpha_g}}.
\end{equation}
Here, $L_{m,g}$ is the linear map that takes $g$ matrix-unit basis elements and returns the output obtained by substituting them for the corresponding occurrences of $G_\Pi$ in the circuit; that is,
\begin{align}
    L_{m,g}\ket{G_\Pi}\!\rangle^{\otimes g} &= L_{m,g}\!\left[\left(\sum_{\alpha_1} g_{\alpha_1} \ket{E_{\alpha_1}}\!\rangle\right) \otimes \cdots \otimes \left(\sum_{\alpha_g} g_{\alpha_g} \ket{E_{\alpha_g}}\!\rangle\right)\right] \\
    & = \sum_{\alpha_1,\ldots,\alpha_g} g_{\alpha_1}\cdots g_{\alpha_g} L_{m,g} \bigl( \ket{E_{\alpha_1}}\!\rangle\otimes\cdots\otimes \ket{E_{\alpha_g}}\!\rangle \bigr) \\
    & = \sum_{\alpha_1,\ldots,\alpha_g} g_{\alpha_1}\cdots g_{\alpha_g} \ket{v_{\alpha_1,\ldots,\alpha_g}}.
\end{align}
Summing over $g$ yields Eq.~\eqref{eq:homogeneous-expansion}, where the factor $e^{ig\theta}$ arises because each of the $g$ occurrences of $G_\Pi$ is multiplied by $e^{i\theta}$.
\end{proof}

The central issue in this proof is to handle the adaptive reuse of the same $G_{\Pi}$.
To address this issue, we introduce the following lemma.
This lemma is a specialization of Lemma~2.11 of~\cite{tang2026conjugatequerieshelp} to random purifications.
\begin{restatable}{lemma}{completelypositivehaarmomentmap}\label{lem:haar-moment-cp}
For every integer $g\geq0$, there exists a completely positive, trace-preserving map
\begin{equation}
    \mathfrak T_g: \mathsf L\!\left((\ket{0}^{\perp})^{\otimes g}\right) \longrightarrow \mathsf L\!\left(( \ket{0}^{\perp} \otimes{\cS^{\perp}})^{\otimes g}\right),
\end{equation}
that depends only on $g,n,k$ and the fixed spaces $ \ket{0}^{\perp} ,{\cS^{\perp}}$ and satisfies the following for every rank-$d_0$ projector $I - \Pi$:
\begin{equation}\label{eq:haar-moment-cp}
    \Gamma_g(I - \Pi) =\mathfrak T_g\!\left(({I-\Pi}^\top)^{\otimes g}\right).
\end{equation}
When $g=0$, both sides are the scalar $1$.
\end{restatable}
\begin{proof}
For $g = 0$, define $\mathfrak{T}_0(1) = 1$.
In the following, suppose that $g \geq 1$.
We use the following result, which we modify with our notation in this appendix.

\begin{lemma}[{\cite[Lemma~2.11]{tang2026conjugatequerieshelp}}]\label{lem:random-purification}
    For integers $g \geq 1$ and $1 \leq d_0 \leq n$, there exists a completely positive, trace-preserving map
    \begin{equation}
        \mathfrak T_g : \mathsf L\!\left((\mathbb{C}^n)^{\otimes g}\right) 
        \longrightarrow
        \mathsf L\!\left((\mathbb{C}^n \otimes \mathbb{C}^{d_0})^{\otimes g}\right),
    \end{equation}
    depending only on $g,n$, and $d_0$, such that, for all mixed states $\rho$ in $\mathbb{C}^n$ with rank at most $d_0$,
    \begin{equation}
        \mathfrak T_g(\rho^{\otimes g}) = \E_{\ket{\varrho}}[(\ketbra{\varrho}{\varrho})^{\otimes g}],
    \end{equation}
    where $\ket{\varrho} \in \mathbb{C}^n \otimes \mathbb{C}^{d_0}$ is sampled uniformly from the purifications of $\rho$, i.e., $\Tr_{\mathbb{C}^{d_0}}[\ketbra{\varrho}{\varrho}] = \rho$.
    The same purification is used in all $g$ tensor factors.
\end{lemma}

    Identify $\mathbb{C}^n \simeq \ket{0}^{\perp}$ and $\mathbb{C}^{d_0} \simeq \mathcal S^{\perp}$, and set
    \begin{equation}
        \rho \coloneqq \frac{I - \Pi^\top}{d_0}.
    \end{equation}
    Since $I - \Pi$ is a rank-$d_0$ projector, $\rho$ is a mixed state of $\rank$ $d_0$.
    From Eqs.~\eqref{eq:choi-vector} and~\eqref{eq:input-output space of G},
    \begin{equation}
        \Tr_{\mathcal{S}^\perp}\left[\frac{J(G_\Pi)}{d_0}\right] = \frac{(G_\Pi^\dagger G_\Pi)^\top}{d_0} = \rho.
    \end{equation}
    Therefore, $\ket{G_\Pi}\!\rangle/{\sqrt{d_0}}$ is a normalized purification of $\rho$.
    Also, the Haar distribution of $G_\Pi$ produces the uniform distribution over these purifications.

    By using Lemma~\ref{lem:random-purification}, we obtain
    \begin{align}
        \mathfrak T_g\left(\left(\frac{I - \Pi^\top}{d_0}\right)^{\otimes g}\right) &= \E_{G_\Pi}\left[\left(\frac{J(G_\Pi)}{d_0}\right)^{\otimes g}\right] \\
        &= \frac{\Gamma_g(I - \Pi)}{d_0^g}.
    \end{align}
    Thus, the map $\mathfrak{T}_g$ is independent of $\Pi$, completely positive, and trace-preserving.
    By linearity, this completes the proof.
\end{proof}

In the following lemma, let $\cH$ be the next query register to be input to the oracle, $\cM$ be all other work registers, and $\cP$ be the space on which the positive operator $X$ obtained so far acts.
The linear map
\[
L:\cP\longrightarrow\cH\otimes\cM
\]
represents the preceding circuit that generates the state of the next query register and the work registers from the state immediately after a query.
Also, on $ \ket{0}^{\perp} $, extend each projector $\Pi$ by zero on $\operatorname{span}\{\ket{0}\}$, and denote the resulting operator on $\cH$ by $\Pi_0$.

\begin{lemma}\label{lem:positive-pairing}
Let $\cP$ be a finite-dimensional Hilbert space, and consider a linear map
\begin{equation}
    L:\cP\longrightarrow\cH \otimes\cM.
\end{equation}
There exists a positive operator
\begin{equation}
    Z_L\succeq0
    \quad\text{acting on}\quad
    \cH \otimes\cP
\end{equation}
such that, for every operator on $\cH$ satisfying $E\succeq0$ and every operator on $\cP$ satisfying $X\succeq0$, we have
\begin{equation}\label{eq:positive-pairing}
    \Tr\!\left[(E\otimes I_\cM)LXL^\dagger\right]
    =\Tr\!\left[Z_L(E^\top\otimes X)\right].
\end{equation}
\end{lemma}

\begin{proof}
Define the completely positive map
\begin{equation}
    \Theta_L:
    \mathsf L(\cH)
    \longrightarrow
    \mathsf L(\cP),
    \qquad
    \Theta_L(E)=L^\dagger(E\otimes I_\cM)L.
\end{equation}
Since $\Theta_L$ is completely positive, its Choi operator satisfies $J(\Theta_L) \succeq 0$.
By Eq.~\eqref{eq:choi-pairing}, we obtain
\begin{equation}
    \Tr[\Theta_L(E)X]
    =\Tr[J(\Theta_L)(E^\top\otimes X)],
\end{equation}
Setting $Z_L \coloneqq J(\Theta_L)$ and using cyclicity of the trace yields Eq.~\eqref{eq:positive-pairing}.
\end{proof}

The following lemma bounds a tensor moment involving a Haar-random projector.
Its proof combines the first-moment identity $\E[\Pi] = (k/n) I$ with an operator union bound for projectors acting on distinct tensor factors.
\begin{lemma}\label{lem:haar-projector}
Let $\Pi\sim\Gr(k,n)$.
Let the additional factor for the next query be $\cH=\operatorname{span}\{\ket{0}\}\oplus\C^n$, and extend $\Pi$ by zero on $\ket{0}$, denoting the extension by $\Pi_0$.
Denote the identity operator on $\cH$ by $I_0$. If $gk<n$, then
\begin{equation}\label{eq:haar-projector-ineq}
    \int \Pi_0^\top\otimes({I - \Pi}^\top)^{\otimes g}\,d\Pi
    \preceq
    \frac{k}{n-gk}
    I_0\otimes
    \int({I - \Pi}^\top)^{\otimes g}\,d\Pi.
\end{equation}
When $g=0$, the tensor power is the scalar identity.
\end{lemma}

\begin{proof}
Set
\begin{equation}
    R_g
    =\int({I - \Pi}^\top)^{\otimes g}\,d\Pi.
\end{equation}
For $j\in[g]$, on the $j$th tensor factor, denote the action of $\Pi^\top$ by $(\Pi^\top)^{(j)}$. These projectors commute because they act on different factors. Thus, the scalar inequality
\begin{equation}
    1-\prod_{j=1}^g(1-x_j)\leq\sum_{j=1}^g x_j,
    \qquad
    x_j\in\{0,1\},
\end{equation}
extends to the operator inequality
\begin{equation}
    I-({I - \Pi}^\top)^{\otimes g}
    \preceq
    \sum_{j=1}^g(\Pi^\top)^{(j)}.
\end{equation}
By Haar invariance,
\begin{equation}
    \int\Pi^\top\,d\Pi
    =\left(\int\Pi\,d\Pi\right)^\top
    =\frac{k}{n}I.
\end{equation}
Therefore, averaging the above operator inequality gives
\begin{equation}\label{eq:Rg-lower}
    R_g
    \succeq
    \left(1-\frac{gk}{n}\right)I
    =\frac{n-gk}{n}I.
\end{equation}
Furthermore,
\begin{align}
    \int \Pi_0^\top\otimes({I - \Pi}^\top)^{\otimes g}\,d\Pi
    &\preceq
      \int \Pi_0^\top\otimes I\,d\Pi \\
    &\preceq\frac{k}{n}I_0\otimes I \\
    &\preceq
      \frac{k}{n-gk}I_0\otimes R_g,
\end{align}
where the last inequality follows from Eq.~\eqref{eq:Rg-lower}. This proves Eq.~\eqref{eq:haar-projector-ineq}.
\end{proof}

Using these lemmas, we now prove Lemma~\ref{lem:first-B-prob}.
\begin{proof}[Proof of Lemma~\ref{lem:first-B-prob}]
For fixed $\Pi$, the probability weight at the $m+1$-th query of obtaining $B$ is obtained by testing the state immediately before the query with $\Pi_0$ on the query register.

For fixed $\Pi$, $\theta$, and $G_\Pi$, consider the first $m$ queries, at each of which $U^{(0)}_{\theta,\Pi,G_\Pi}$ is applied, and denote the resulting subnormalized pure state by $\ket{\Phi_m(\theta,G_\Pi)}$.
By Lemma~\ref{lem:homogeneous-network}, we have
\begin{equation}\label{eq: branch-expansion-without-B-part}
    \ket{\Phi_m(\theta,G_\Pi, \Pi)} =\sum_{g=0}^m e^{ig\theta} L_{m,g}\ket{G_\Pi}\!\rangle^{\otimes g}.
\end{equation}
Here, $L_{m,g}$ is a fixed linear map independent of $\Pi$ and $G_\Pi$, with output space $\cH\otimes\cM$.
Averaging over the common phase $\theta$ eliminates cross terms between different values of $g$.
\begin{align}
    \E_\theta[\proj{\Phi_m(\theta,G_\Pi,\Pi)}] &= \E_\theta\left[\sum_{g,h=0}^m e^{i(g-h)\theta} L_{m,g}\ket{G_\Pi}\!\rangle^{\otimes g} \langle\!\bra{G_\Pi}^{\otimes h} L_{m,h}^\dagger\right] \\
    &= \sum_{g=0}^m
      L_{m,g}\ket{G_\Pi}\!\rangle^{\otimes g}\langle\!\bra{G_\Pi}^{\otimes g}L_{m,g}^\dagger \\
    &= \sum_{g=0}^m
      L_{m,g}J(G_\Pi)^{\otimes g}L_{m,g}^\dagger.
\end{align}
Averaging further over $G_\Pi$ expresses the $g$th component in terms of $\Gamma_g(I-\Pi)$:
\begin{align}
\E_{G_\Pi}\left[\sum_{g=0}^m
      L_{m,g}J(G_\Pi)^{\otimes g}L_{m,g}^\dagger\right] &= \sum_{g=0}^m L_{m,g}\Gamma_g({I - \Pi})L_{m,g}^\dagger \\
    &= \sum_{g=0}^m
      L_{m,g} \mathfrak T_g(({I - \Pi}^\top)^{\otimes g})L_{m,g}^\dagger,
\end{align}
Here, the first equality follows from Eq.~\eqref{eq:Gamma-g-def}, and the second from Lemma~\ref{lem:haar-moment-cp}.
Apply Lemma~\ref{lem:positive-pairing} to $L_{m,g}$, and denote the resulting positive operator by $Z_{L_{m,g}}$.
Since $Z_{L_{m,g}}$ depends only on $L_{m,g}$, it is independent of $\Pi$.
Thus, averaging also over $\Pi$, the probability of the first $B$, denoted by $p_{m+1}^{\mathrm{first}}$, can be written as
\begin{align}
    p_{m+1}^{\mathrm{first}} & = \int\Tr\!\left[
        (\Pi_0\otimes I_\cM)
         \sum_{g=0}^m
      L_{m,g} \mathfrak T_g(({I - \Pi}^\top)^{\otimes g})L_{m,g}^\dagger
      \right] d\Pi\\
    & = \sum_{g=0}^m \Tr\!\left[
        \int (\Pi_0\otimes I_\cM)
      L_{m,g} \mathfrak T_g(({I - \Pi}^\top)^{\otimes g})L_{m,g}^\dagger \,d\Pi
      \right] \\
    &=\sum_{g=0}^m
      \Tr\!\left[
        Z_{L_{m,g}}
        \int \Pi_0^\top\otimes\mathfrak T_g(({I - \Pi}^\top)^{\otimes g})\,d\Pi
      \right] \\
    &=\sum_{g=0}^m
      \Tr\!\left[
        (\id\otimes\mathfrak T_g^\dagger)(Z_{L_{m,g}})
        \int \Pi_0^\top\otimes({I - \Pi}^\top)^{\otimes g}\,d\Pi
      \right],
      \label{eq:first-B-positive-pairing}
\end{align}
In the last equality, we used the definition of the Hilbert--Schmidt adjoint.
Since $\mathfrak T_g$ is completely positive, its adjoint $\mathfrak T_g^\dagger$ is also completely positive, and
\begin{equation}
    (\id\otimes\mathfrak T_g^\dagger)(Z_{L_{m,g}})\succeq0.
\end{equation}
Therefore, we can apply Lemma~\ref{lem:haar-projector} to each term in Eq.~\eqref{eq:first-B-positive-pairing}.
Since $g\leq m$ implies $k/(n-gk)\leq k/(n-mk)$, we have
\begin{align}
    p_{m+1}^{\mathrm{first}}
    &\leq
      \frac{k}{n-mk}
      \sum_{g=0}^m
      \Tr\!\left[
        Z_{L_{m,g}}
        \int I_0\otimes\mathfrak T_g(({I - \Pi}^\top)^{\otimes g})\,d\Pi
      \right] \\
    &=\frac{k}{n-mk}s_m.
\end{align}
In the last equality, we replaced the next-query test $\Pi_0$ by $I_0$ in Lemma~\ref{lem:positive-pairing}.
This completes the proof of Lemma~\ref{lem:first-B-prob}.
\end{proof}

\end{document}